\documentclass[journal]{IEEEtran}
\usepackage{cite}
\usepackage{amsmath,amssymb,amsfonts,amsthm}
\usepackage{algorithmic}
\usepackage{physics}
 \usepackage{graphicx}
\usepackage{textcomp}
\usepackage[T1]{fontenc}
\usepackage{caption}
\usepackage{subcaption}
\usepackage{epstopdf}
\usepackage{overpic}
\usepackage{array}
\usepackage{color}
\usepackage{url}
\usepackage{bbm}
\usepackage{enumerate}
\usepackage[shortlabels]{enumitem}

\usepackage[none]{hyphenat}
\usepackage{hyperref}
\usepackage{lipsum}

\def\BibTeX{{\rm B\kern-.05em{\sc i\kern-.025em b}\kern-.08em
    T\kern-.1667em\lower.7ex\hbox{E}\kern-.125emX}}
    
\usepackage{color}
\usepackage[dvipsnames,svgnames]{xcolor}
\usepackage[textwidth=30mm]{todonotes}

\def\Htran{\mbox{\tiny $\mathrm{H}$}}
\def\Ttran{\mbox{\tiny $\mathrm{T}$}}
\def\CN{\mathcal{N}_{\mathbb{C}}}

\def\Real{\mathbb{R}}
\def\Complex{\mathbb{C}}

\def\Ex{\mathbb{E}}
\def\sinc{\mathrm{sinc}}

\def\kron{\otimes}
\def\diag{\mathrm{diag}}

\def\imagunit{\mathsf{j}} 

\def\ktx{\boldsymbol{\kappa}}
\def\krx{\vect{k}}

\newcommand{\vect}[1]{{\bf{#1}}}

\theoremstyle{plain}
\newtheorem{theorem}{Theorem}

\newtheorem{lemma}{Lemma}
\newtheorem{corollary}{Corollary}
\newtheorem{remark}{Remark}

\newtheorem{assumption}{Assumption}

\IEEEoverridecommandlockouts

\begin{document}

 \title{\Huge{Fourier Plane-Wave Series Expansion for Holographic MIMO Communications}}

\author{
\IEEEauthorblockN{Andrea Pizzo, \emph{Member, IEEE}, Luca Sanguinetti, \emph{Senior Member, IEEE}, \\Thomas L. Marzetta, \emph{Life Fellow, IEEE}\vspace{-0.5cm}
\thanks{
\newline \indent Part of this work was presented at the Asilomar Conference on signals, Systems and Computers, Pacific Grove, CA, 2021 \cite{PizzoASILOMAR20}. A.~Pizzo is with the Department of Information and Communication Technologies, Universitat Pompeu Fabra, 08018 Barcelona, Spain (andrea.pizzo@upf.edu).
L.~Sanguinetti is with the Dipartimento di Ingegneria dell'Informazione, University of Pisa, 56122 Pisa, Italy (luca.sanguinetti@unipi.it). T. L. Marzetta is with the Department of Electrical and Computer Engineering, Tandon School of Engineering, 11201 Brooklyn, NY (tom.marzetta@nyu.edu). 
\newline \indent L. Sanguinetti was partially supported by the Italian Ministry of Education and Research (MIUR) in the framework of the CrossLab project (Departments of Excellence).
}
}}

\maketitle

\begin{abstract}
Imagine a MIMO communication system that fully exploits the propagation characteristics offered by an electromagnetic channel and ultimately approaches the limits imposed by wireless communications.
This is the concept of Holographic MIMO communications. Accurate and tractable channel modeling is critical to understanding its full potential.
Classical stochastic models used by communications theorists are derived under the electromagnetic far-field assumption, i.e. planar wave approximation over the array. However, such assumption breaks down when {electromagnetically} large (compared to the wavelength) antenna arrays are considered. 
In this paper, we start from the first principles of wave propagation and provide a Fourier plane-wave series expansion of the channel response, which fully captures the essence of electromagnetic propagation in arbitrary scattering and is also valid in the (radiative) near-field. The expansion is based on the Fourier spectral representation and has an intuitive physical interpretation, as it statistically describes the angular coupling between source and receiver.
When discretized {uniformly}, it leads to a low-rank semi-unitarily equivalent approximation of the electromagnetic channel in the angular domain.
The developed channel model is used to compute the ergodic capacity of a point-to-point Holographic MIMO system with different degrees of channel state information.

\end{abstract}




\section{Introduction}

Communication theorists are constantly looking for new technologies to increase the information rate and reliability of wireless communications. Chief among the technologies that blossomed into major advances is the multiple antenna technology, whose latest instantiation, i.e., Massive MIMO (multiple-input multiple-output), became a reality in 5G~\cite{BJORNSON20193}. 
Inspired by the potential benefits of Massive MIMO with more and more antennas, most of the new research directions envision the use of dense { and electromagnetically large} (compared to the wavelength $\lambda$) antenna arrays, and are taking place under different names, e.g., Holographic MIMO~\cite{PizzoJSAC20}, large intelligent surfaces \cite{Rusek2018}, and reconfigurable intelligent surfaces \cite{RIS}. 
Particularly, the Holographic MIMO concept refers to a MIMO system which is designed to fully exploit the propagation characteristics offered
by an electromagnetic channel; this definition comes from the \emph{holographic} term, which dates back to the ancient greek and literally means ``describe everything'' \cite{DardariHolographic}.

Realistic design and performance assessment of { electromagnetically large multiple} antenna technologies require accurate and tractable channel models {for the wave propagation}. 
Deterministic models (e.g., based on ray tracing) achieve the highest accuracy as they provide accurate predictions of signal propagation in a given environment \cite{Sarkar2003}. However, they rely on numerical electromagnetic solvers of the Maxwell's equations, and hence, they are too site-specific. Stochastic models are the most desirable for communication theorists to work with as they are {representative of a class of environments with common propagation properties} \cite{MarzettaISIT}.
Physically meaningful stochastic models are based on a channel expansion in terms of {plane waves or spherical waves}, as they provide a natural eigensolution of the {wave} equation \cite{PizzoJSAC20}. 
Unlike models based on a spherical wave expansion of the channel (e.g.,~\cite{Gustafsson2006,Gustafsson2009,Gustafsson2010}), channel models that are based on plane waves allow to treat radio wave propagation as a linear system {by using Fourier theory} and without the recourse to special functions \cite{PizzoJSAC20,PizzoIT21}. In addition, {plane-wave models} are particularly useful as they decouple scattering conditions from array characteristics \cite{PoonCapacity,Sayeed2002}.

Unfortunately, the use of plane-wave based models in wireless research is generally confined to the far-field (Fraunhofer) regime only, where wavefronts are approximated as locally planar over the entire array \cite{PoonCapacity,Sayeed2002,Veeravalli,PoonDoF,Kennedy2007,Pollock,LucaBook}. An example is given by the virtual channel representation pioneered in~\cite{Sayeed2002,Veeravalli}. 
Notice that the use of a far-field model in the near-field (radiative Fresnel) regime would lead to magnitude and phase errors at the receiver due to a non-negligible curvature of the incoming wavefronts (e.g., at {the Fraunhofer distance} ${R=2 L^2/\lambda}$ we have a maximum phase error of $\pi/8$ across {an array of size $L$ \cite[Sec.~2.2]{BalanisBook}).
Notice that electromagnetically large arrays pushes the electromagnetic operating regime from the far-field to the near-field  regime, as they are specified by a lower Fraunhofer distance than traditional antenna arrays~\cite{DardariHolographic,Franceschetti_NearField}. To this end,} Table~\ref{tab:near_field} reports the Fraunhofer distance in meters for arrays of practical size at $3$, $28$, $73$, and $142$~GHz carrier frequencies \cite{xing2021}. As seen, the near-field regime may occur at any frequency for applications with not only short- but also mid-range distances.

Recently, {however,} \cite{PizzoASILOMAR20,PizzoIT21} have brought to the attention of the wireless community that wave propagation can always be formulated in terms of {plane waves} irrespective of the distance between source and receiver (i.e., even in the near-field) and under arbitrary scattering conditions.
This formulation builds upon the fact that every transmitted spherical wave can be decomposed \emph{exactly} into an infinite number of {plane waves}~\cite{Weyl,ChewBook,PlaneWaveBook}.
Upon interaction with the scatterers, another (possibly) infinite number of received {plane waves} is created contributing to the receive field. 
The entire {effect of the} scattering mechanism is embedded into an angular response that maps propagation from every transmit direction to every other receive direction~\cite{Saxon,NietoWolf}. 
An analytically tractable stochastic model for the angular response is obtained by {selecting its entries as being statistically uncorrelated from one direction to another, which implies the field to be spatially stationary in the radiative near-field region~\cite{PizzoIT21,MarzettaNOKIA}.}

\begin{table}[t] 
        \caption{Fraunhofer distances for different array apertures.\vspace{-0cm}}  \label{tab:near_field}
\centering
    \begin{tabular}{|c|c|c|c|c|c|} 
    \hline
    {Maximum size $L$\,[m]} & {$3$~GHz} & {$28$~GHz} & {$73$~GHz}  & {$142$~GHz}    \\     \hline\hline 
     	$0.1$ &  	$-$ & 		$1.9$ & 	$4.9$ &	$9$   \\  \hline
	$0.5$ &  	$5$ & 		$47$ & 	$122$ &	$237$   \\  \hline
      	$1$ &  	$20$   & 		$187$ & 	$487$ &	 $-$   \\  \hline
     	$3$ & 	$180$  & 		$-$ & 	$-$ &	 $-$  \\  \hline
    \end{tabular}
\end{table}

\subsection{Contributions}

We consider wireless communications between two parallel planar arrays in a three-dimensional (3D) arbitrary scattered medium and provide {a continuous description of wireless propagation through an approximated \emph{Fourier plane-wave series expansion} of an electromagnetic random channel.}
The provided model complies with the physics of wave propagation and incorporates spatial correlation effects due to directionality of the field generated by the scattering.
It is based on a discretization of the Fourier spectral representation of stationary electromagnetic random fields provided in~\cite{PizzoIT21}, which asymptotically yields a continuum of uncorrelated and circularly-symmetric, complex-Gaussian random coefficients. For finite $L/\lambda$ values, only a subset of these coefficients carries the essential channel information, thus revealing the quantized nature of the physical world. The variances of these coefficients fully describe the field statistically and have an intuitive physical interpretation, as being the strengths of the angular coupling between the source and receive arrays, which can thus be measured accordingly. They are determined by the joint propagation characteristics at both link ends. {Precisely, the provided channel description can be regarded as the Karhunen-Loeve expansion of a stationary electromagnetic random field as $L/\lambda\gg 1$.}

When discretized {uniformly} at Nyquist's spacing, the Fourier plane-wave series expansion yields a stochastic description of the electromagnetic MIMO channel in which the array geometry and scattering conditions are perfectly separated.
The former is a deterministic effect that changes the domain of representation from spatial to angular (and vice-versa) and is performed by a {double 2D discrete Fourier transform (DFT) operation.} 
The latter is a stochastic effect that accounts for wave propagation under different environments. Notably, the angular domain provides a low-rank semi-unitarily equivalent approximation of the electromagnetic channel.
The decoupling property of the model yields an efficient hybrid structure for the transceiver architecture that accounts for a double 2D DFT operation in the analog stage and enables the design of array configuration and signal processing algorithms separately. Since the spatial domain offers a highly redundant description of the electromagnetic channel, a significant complexity reduction (e.g., channel estimation, optimal signaling, coding) can be achieved. The developed  model is finally used to compute the ergodic capacity for different degrees of channel state information.

We conclude this section by observing that the Fourier plane-wave series expansion derived in this paper differs from the one computed in \cite[Sec.~V]{PizzoJSAC20} in the following aspects: \emph{i}) it encompasses both link ends (i.e., source and receiver) while only the receiver is considered in~\cite{PizzoJSAC20}; \emph{ii}) the analytical framework in Section III.C for the computation of the coupling coefficients applies to an arbitrary configuration of scatterers while only isotropic scattering is considered in~\cite{PizzoJSAC20}; \emph{iii}) it considers the practical case where the source and receiver are composed of a finite number of radiative/sensing elements.

\subsection{Outline of the Paper and Notation}

The remainder of this paper is organized as follows. In Section~\ref{sec:scattering_model}, we briefly review the Fourier plane-wave representation from~\cite{PizzoIT21}. This is used in Section~\ref{sec:Fourier_series} to derive a novel Fourier plane-wave series expansion of an electromagnetic random channel. Suitably discretized, this expansion yields in Section~\ref{sec:MIMO_channel_model} a stochastic description of an electromagnetic MIMO channel. Section~\ref{sec:capacity} uses the provided model to analyze the capacity of the channel and includes numerical results for illustration. Final discussions are drawn in Section~\ref{sec:conclusions}.
  
We use upper (lower) case letters for angular (spatial) entities and boldfaced letters for vectors and matrices. Sets are indicated by calligraphic letters.  For a set $\mathcal{X}$, $|\mathcal{X}|$ and $\mathbbm{1}_{\mathcal{X}}(x)$ are the Lebesgue measure and indicator function.
The notation ${n \sim \CN(0, \sigma^2)}$ stands for a circularly-symmetric complex-Gaussian random variable with variance $\sigma^2$. $\Ex\{\cdot\}$ is the expectation operator. The Hadamard and Kronecker products are $\vect{A} \odot \vect{B}$ and $\vect{A} \otimes \vect{B}$.
We denote $\vect{I}_N$ the $N\times N$ identity matrix and $\diag(\vect{a})$ the diagonal matrix with elements from $\vect{a}$.
$\Real^n$ is the $n$-dimensional space of real-valued numbers.
A general point in $\Real^3$ is described by $\vect{r} = x \hat{\vect{x}} + y \hat{\vect{y}} + z \hat{\vect{z}}$ where $\hat{\vect{x}}$, $\hat{\vect{y}}$, and $\hat{\vect{z}}$ are three orthonormal vectors, and $(x,y,z)$ are its Cartesian coordinates.
The length of $\vect{r}$ is $\|\vect{r}\| = \sqrt{x^2 + y^2 + z^2}$ and $\hat{\vect{r}} = \vect{r}/\|\vect{r}\|$ is the unit vector. 

\section{Preliminaries} \label{sec:scattering_model}

Consider the two parallel and $z$-oriented planar arrays depicted in Fig.~\ref{fig:propagation}, which span the rectangular regions $\mathcal{S} \subset \Real^2$ and $\mathcal{R} \subset \Real^2$ of $xy$-dimensions ${L_{S,x},L_{S,y}}$ and ${L_{R,x},L_{R,y}}$, respectively.\footnote{Volumetric arrays do not offer extra degrees of freedom (DoF) over planar arrays \cite{PizzoSPAWC20}.} The transmit array is equipped with $N_S$ antenna elements while the receive array has $N_R$ antennas.
Wave propagation takes place in the form of monochromatic scalar waves (i.e., with no polarization), at radiation frequency $\omega$ (corresponding to a wavelength $\lambda$), in a 3D scattered, homogeneous, and infinite medium. 
We assume that there is no direct path due to the presence of scatterers, which are homogeneous and made up of arbitrary shape and size. 

\begin{figure}[t!]
    \centering
     \includegraphics[width=1.\columnwidth]{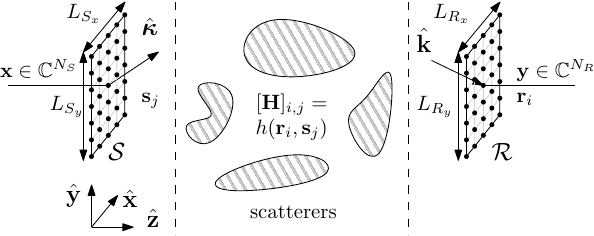} 
       \caption{{\small MIMO communications system under arbitrary scattering. }
       }
   \label{fig:propagation}
\end{figure}

\subsection{Non Line-of-Sight MIMO communications}

At any symbol time, the MIMO system in Fig.~\ref{fig:propagation} is described by the following discrete-space model (e.g.,~\cite{heath_lozano_2018}):
\begin{equation}\label{eq:MIMO_channel}
{\bf y} = {\bf H}{\bf x} + {\bf n}
\end{equation} 
where ${{\bf y}\in \mathbb{C}^{N_R}}$ and ${{\bf x}\in \mathbb{C}^{N_S}}$ denote the received and transmitted signal vectors, respectively. Also, ${{\bf n}\in \mathbb{C}^{N_R}}$ accounts for thermal noise that is distributed as ${\vect{n} \sim\CN({\bf 0},\sigma^2 {\bf I}_{N_R})}$. Here, the entry $[{\bf H}]_{ij}$ represents the propagation coefficient between the $j$th transmitting antenna located at point $\vect{s}_j$ and the $i$th receiving antenna located at point $\vect{r}_i$. In non line-of-sight communications, the entries $[\vect{H}]_{ij}$ are typically modeled as a stationary sequence of circularly-symmetric, complex-Gaussian and correlated random variables. As a consequence, ${\vect{H} \in \Complex^{N_R \times N_S}}$ is a correlated Rayleigh fading channel matrix, uniquely described by its spatial correlation matrix ${{\bf R}\in \Complex^{N_R N_S \times N_RN_S}}$ as
\begin{equation}\label{full_correlation_matrix}
{\bf R}= \mathbb{E}\{{\rm vec({\bf H})}{\rm vec}({\bf H})^{\Htran}\}.
\end{equation}
The classical approach is to develop physically-meaningful models for ${\bf R}$ from which realizations of $\vect{H}$ are then obtained. Differently, this paper builds upon \cite{PizzoIT21} that, starting from first electromagnetic principles of wave propagation, provides a \emph{Fourier plane-wave representation} of $h({\bf r},{\bf s})$, namely the random channel response at point ${\bf r}$ due to a unit impulse (point source) at point ${\bf s}$. The MIMO channel matrix in \eqref{eq:MIMO_channel} is obtained by sampling $h({\bf r},{\bf s})$ at ${\bf s}_j = [{s}_{x_j},{s}_{y_j}, s_z]^T$ and ${\bf r}_i = [{r}_{x_i},{r}_{y_i}, r_z]^T$ for $j=1,\ldots,N_S$ and $i=1,\ldots,N_R$ as
\begin{equation}\label{eq:MIMO_entries}
[{\bf H}]_{ij} = h({\bf r}_i,{\bf s}_j)
\end{equation} 
for any fixed pair $(r_z,s_z)$.
Notice that the plane-wave representation of $h({\bf r},{\bf s})$ is asymptotically exact as $\min(L_{S,x},L_{S,y})/\lambda \to \infty$ and $\min(L_{R,x},L_{R,y})/\lambda \to \infty$ jointly \cite{PizzoIT21}. The key results from \cite{PizzoIT21} are reviewed next as they are instrumental for Section~\ref{sec:Fourier_series}, where a novel \emph{Fourier plane-wave series expansion} is derived that well approximates $h({\bf r},{\bf s})$ within $(r_x,r_y) \in \mathcal{R}$ and $(s_x,s_y) \in \mathcal{S}$ when $\min(L_{S,x},L_{S,y})/\lambda \gg 1$ and $\min(L_{R,x},L_{R,y})/\lambda \gg 1$ jointly.

\begin{remark}
The wireless channel is composed of large-scale fading and small-scale fading. The former occurs on a larger scale --- a few hundred wavelengths --- and is due to pathloss, shadowing, and antenna gains, while the latter is a microscopic effect caused by small variations in the propagation. 
If the array size at both ends does not exceed the size of the local scattering neighbourhood, the two components can be modelled independently (e.g.,~\cite[Sec.~3.6]{heath_lozano_2018}). 
This paper only considers the small-scale fading. Any large-scale fading model can be applied verbatim.  
\end{remark}

\subsection{Fourier plane-wave representation of electromagnetic channels}

Once the reactive propagation mechanisms taking place in the proximity of source and scatterers (i.e., at a distance of few wavelengths) are excluded, the channel response $h(\vect{r},\vect{s})$ measured between two infinitely large planar arrays can be modeled as a \emph{spatially-stationary} electromagnetic random field~\cite{PizzoIT21}. For this class of channels, we can always find an \emph{exact} statistical representation of $h(\vect{r},\vect{s})$ in terms of {plane waves} that is given by the four-dimensional (4D) Fourier plane-wave representation~\cite{PizzoIT21}.
In particular, when the scatterers are confined entirely within the region separating source and receiver, 
\begin{align}\notag
& h(\vect{r},\vect{s})  =   \frac{1}{(2\pi)^2} \iiiint_{\mathcal{D}(\kappa)\times \mathcal{D}(\kappa)} a_R(k_x,k_y,\vect{r}) \\& \hspace{.7cm} \times
H_a(k_x,k_y,\kappa_x,\kappa_y)  a_S(\kappa_x,\kappa_y,\vect{s}) \, dk_xdk_y d\kappa_xd\kappa_y\label{Fourier_planewave}
\end{align}
and it is thus decomposed into three terms.
The first term $a_S(\kappa_x,\kappa_y,\vect{s})$ is the \emph{source response} that maps the impulsive excitation current at point $\vect{s}$ to the {transmit} propagation direction $\hat \ktx = \ktx/||\ktx||$ of the {transmitted} field. The second term $a_R(k_x,k_y,\vect{r})$ is the \emph{receive response} that maps the receive propagation direction $\hat \krx = \krx/||\krx||$ of the receive field to the induced current at point $\vect{r}$. They are defined as
\begin{align} \label{plane-wave-tx}
a_S(\kappa_x,\kappa_y,\vect{s}) & = e^{-\imagunit  \ktx ^{\Ttran} \vect{s}} = e^{-\imagunit \big( \kappa_x s_x + \kappa_y s_y + \gamma(\kappa_x,\kappa_y) s_z \big)} \\ \label{plane-wave-rx}
a_R(k_x,k_y,\vect{r}) & = e^{\imagunit  \krx ^{\Ttran} \vect{r}} = e^{\imagunit \big( k_x r_x + k_y r_y + \gamma(k_x,k_y) r_z\big)} 
\end{align}
where $\ktx = \kappa_x \hat{\vect{x}} + \kappa_y \hat{\vect{y}} + \gamma(\kappa_x,\kappa_y) \hat{\vect{z}}$ and $\krx = k_x \hat{\vect{x}} + k_y \hat{\vect{y}} + \gamma(k_x,k_y) \hat{\vect{z}}$ are the corresponding wave vectors with
\begin{equation} \label{gamma}
\gamma(k_x,k_y) = \sqrt{\kappa^2 - k_x^2 - k_y^2}
\end{equation}
given $\kappa=2\pi/\lambda$ as the wavenumber. The integration region in \eqref{Fourier_planewave} is limited to the support
\begin{equation}  \label{disk_T}
\mathcal{D}(\kappa) = \{ (k_x,k_y)\in\Real^2 : k_x^2 + k_y^2 \le \kappa^2\}
\end{equation}
given by a disk of radius $\kappa$. As a result, $\gamma(\cdot,\cdot)$ is always real-valued and the representation in \eqref{Fourier_planewave} involves propagating {plane waves} only.
Note that this is due to the spatial stationarity of $h(\vect{r},\vect{s})$ that reveals the low-pass filter behavior of the electromagnetic channel.
The third term $H_a(k_x,k_y,\kappa_x,\kappa_y)$ in~\eqref{Fourier_planewave} is the \emph{angular response} that maps every source direction $\hat\ktx$ onto every receive direction $\hat\krx$. 
Its statistical structure is given in the following theorem.

\begin{theorem} \cite{PizzoIT21} \label{th:stationary}
If $h(\vect{r},\vect{s})$ is a spatially-stationary, circularly-symmetric and complex-Gaussian random field, the angular response $H_a(k_x,k_y,\kappa_x,\kappa_y)$ is of the form 
\begin{equation} \label{angular_response_stationary}
\!\!H_a(k_x,k_y,\kappa_x,\kappa_y) =   \frac{A(k_x,k_y,\kappa_x,\kappa_y)W(k_x,k_y,\kappa_x,\kappa_y)}{\gamma^{1/2}(k_x,k_y) \gamma^{1/2}(\kappa_x,\kappa_y)}  
\end{equation}
where $A(k_x,k_y,\kappa_x,\kappa_y)$ is an arbitrary non-negative function (called spectral factor) and $W(k_x,k_y,\kappa_x,\kappa_y)$ is a collection of unit-variance, independent and identically distributed (i.i.d.) circularly-symmetric and complex-Gaussian random variables, i.e., $W(k_x,k_y,\kappa_x,\kappa_y) \sim \CN(0,1)$.
\end{theorem}

Plugging~\eqref{angular_response_stationary} into~\eqref{Fourier_planewave} generates a stationary random field that converges in the {mean-squared-error} sense to $h(\vect{r},\vect{s})$ for any channel with bounded spectral factor $A(k_x,k_y,\kappa_x,\kappa_y)$ \cite{PizzoIT21}.  It provides a second-order characterization of the channel response in terms of statistically independent complex-Gaussian random coefficients.
In fact, \eqref{Fourier_planewave} is directly connected to the Fourier spectral representation of a spatial random field.

\begin{remark} \label{remark_decoupling} 
The series expansion of the channel in~\eqref{Fourier_planewave} leads to the decoupling of array geometry and scattering, in line with previous works on channel modelling that rely on plane-wave decompositions, e.g., 
\cite{PoonCapacity,Sayeed2002,Veeravalli,PoonDoF,Kennedy2007,Pollock,LucaBook}. While the former is a deterministic effect represented by the source and receive responses, the latter is a stochastic effect that is entirely embedded into the angular response in \eqref{angular_response_stationary}.
In Section~\ref{sec:MIMO_channel_model}, this property will lead to a MIMO channel model whose spatial correlation matrix has a decoupled structure, thereby enabling the design of array configurations and signal processing algorithms separately.
\end{remark}

\subsection{Physical considerations}

The angular response $H_a(k_x,k_y,\kappa_x,\kappa_y)$ describes the channel coupling between every pair of source $\hat\ktx$ and receive $\hat\krx$ propagation directions. We may rewrite~\eqref{angular_response_stationary} as
\begin{align}  \label{scattering_response_nlos}
 \!\!\!\!\!\! H_a(k_x,k_y,\kappa_x,\kappa_y) \!=\!  S^{1/2}(k_x,k_y,\kappa_x,\kappa_y) W(k_x,k_y,\kappa_x,\kappa_y)  \!\!
\end{align}
where $S(k_x,k_y,\kappa_x,\kappa_y)$ is a non-negative function defined as
\begin{equation} \label{psd_4d}
\!\!S(k_x,k_y,\kappa_x,\kappa_y)  =\frac{A^2(k_x,k_y,\kappa_x,\kappa_y)}{\gamma(k_x,k_y) \gamma(\kappa_x,\kappa_y)}.
\end{equation}
Plugging~\eqref{scattering_response_nlos} into~\eqref{Fourier_planewave} the average channel power $P = \Ex\{|h(\vect{r},\vect{s})|^2\} <\infty$ is~\cite{PizzoIT21}:
 \begin{align}\label{power}
P & = \frac{1}{(2\pi)^4} \iiiint_{-\infty}^\infty \! S(k_x,k_y,\kappa_x,\kappa_y)  \, d\kappa_x d\kappa_y dk_x dk_y
\end{align}
where $S(k_x,k_y,\kappa_x,\kappa_y)$ represents the \emph{bandlimited} 4D power spectral density of $h(\vect{r},\vect{s})$.
If $h(\vect{r},\vect{s})$ is assumed to have unit average power, then $S(k_x,k_y,\kappa_x,\kappa_y)$ can be regarded as a \emph{continuous angular power distribution} of the channel, which specifies the power transfer between every pair of {transmit} $\hat\ktx$ and receive $\hat\krx$ propagation directions, on average. From~\eqref{psd_4d}, it follows that it is fully described by the spectral factor $A(k_x,k_y,\kappa_x,\kappa_y)$ that physically accounts for the angular selectivity of the scattering. This function uniquely parametrizes the channel model and should be chosen to fit a prescribed class of propagation environments (e.g., through channel measurements). In the special case of isotropic scattering, the spectral factor is constant over its domain, i.e., $A(k_x,k_y,\kappa_x,\kappa_y) = A(\kappa)$, as the transfer of power is uniformly distributed over all propagation directions \cite{PizzoJSAC20,PizzoIT21}.
Under non-isotropic scattering, the spectral factor is not constant and the bandwidth of $h(\vect{r},\vect{s})$ is determined by the support of $S(k_x,k_y,\kappa_x,\kappa_y)$ in~\eqref{psd_4d}, as summarized next.

\begin{corollary}\cite{PizzoJSAC20,PizzoIT21,PizzoTSP21} \label{th:bandlimited}
$h(\vect{r},\vect{s})$ is bandlimited with maximum circular bandwidth $|\mathcal{D}(\kappa)| =\pi\kappa^2$, achieved under isotropic scattering conditions.
\end{corollary}

As derived in~\cite[Lemma~2]{PizzoJSAC20}, under isotropic propagation conditions, the correlation function between two antennas at a distance $r$ yields the well-known Clarke's isotropic $\sinc(2r/\lambda)$ correlation. This shows that~\eqref{Fourier_planewave} embraces existing models and proves its asymptotic validity, since the Clarke's model is exact under isotropic propagation~\cite{PizzoJSAC20}.

\section{Fourier plane-wave series of stochastic electromagnetic channels} \label{sec:Fourier_series}

For any fixed pair $(r_z,s_z)$, the source and received {plane waves} in \eqref{plane-wave-tx} and \eqref{plane-wave-rx} correspond to two phase-shifted versions of {two-dimensional (2D)} spatial-frequency Fourier harmonics.
Intuitively, the transition from the Fourier plane-wave representation to a Fourier plane-wave series expansion is analogous to the Fourier integral-Fourier series transition for time-domain signals.
This is the main result of this section, which is summarized in Theorem~\ref{th:series_expansion} below and provides us with an approximation of \eqref{Fourier_planewave} for planar arrays of finite extent. The approximation is accurate in the large array regime, as summarized next.

\begin{assumption}\label{Assumption1}Arrays are {electromagnetically large} such that $\min(L_{S,x},L_{S,y})/\lambda\gg 1$ and $\min(L_{R,x},L_{R,y})/\lambda\gg 1$.
\end{assumption}
The above assumption does not require the arrays to be ``physically large'', but rather their normalized size (compared to the wavelength). This is analogous to the Nyquist-Shannon sampling theorem for bandlimited waveform channels $h(t)$ of bandwidth $B$, observed over time interval $[0,T]$ (e.g.,~\cite{FranceschettiBook}). For this class of channels, we can approximate $h(t)$ as a linear combination of a finite number of cardinal basis functions with coefficients collected inside $[0,T]$ and equally spaced by $1/2B$. The approximation error within $[0,T]$ becomes negligible as ${B T \gg 1}$ and is zero at ${B T\to\infty}$.

\subsection{Main result}

With a slight abuse of notation, we call 
\begin{align} \label{plane_wave_discrete_tx}
\!\!a_S(m_x,m_y,\vect{s}) & = e^{\!\!-\imagunit \Big(\!\frac{2\pi}{L_{S,x}} m_x s_x + \frac{2 \pi}{L_{S,y}} m_ys_y + \gamma_S(m_x,m_y) s_z\! \Big)} \!\!\\ \label{plane_wave_discrete_rx}
a_R(\ell_x,\ell_y,\vect{r}) & = e^{\imagunit \Big(\frac{2\pi}{L_{R,x}} \ell_xr_x + \frac{2 \pi}{L_{R,y}} \ell_yr_y + \gamma_R(\ell_x,\ell_y) r_z\Big)}
\end{align}
the discretized {plane waves} obtained by evaluating $(\kappa_x,\kappa_y)$ at $({2\pi m_x}/L_{S,x},{2\pi m_y}/L_{S,y})$ and $(k_x,k_y)$ at $({2\pi \ell_x}/L_{R,x},{2\pi \ell_y}/L_{R,y})$, respectively. Consequently, 
 \begin{align} \label{gamma_s}
 \gamma_S(m_x,m_y)  &= \gamma\left(\frac{2\pi}{L_{S,x}} m_x,\frac{2\pi}{L_{S,y}} m_y\right) \\\label{gamma_r}
 \gamma_R(\ell_x,\ell_y) &= \gamma\left(\frac{2\pi}{L_{R,x}} \ell_x,\frac{2\pi}{L_{R,y}} \ell_y\right)
 \end{align}
 where $\gamma(\cdot,\cdot)$ is given by \eqref{gamma}.
Since the angular response $H_a(k_x,k_y,\kappa_x,\kappa_y)$ is non-zero only within the support $(k_x,k_y,\kappa_x,\kappa_y)\in\mathcal{D}(\kappa)\times\mathcal{D}(\kappa)$, the discretized {plane waves} in \eqref{plane_wave_discrete_tx} and \eqref{plane_wave_discrete_rx} are defined within the lattice ellipses (e.g. \cite[Fig. 1]{PizzoSPAWC20})
\begin{align}\label{epsilon_s}
\hspace{-0.4cm}\mathcal{E}_S &= \left\{(m_x,m_y)\!\in\!\mathbb{Z}^2 \!:\! \left(\frac{m_x \lambda}{L_{S,x}}\right)^2 \!+ \!\left(\frac{m_y \lambda}{L_{S,y}}\right)^2 \le 1\right\}\!\!
\\\label{epsilon_r}
\hspace{-0.4cm}\mathcal{E}_R &= \left\{(\ell_x,\ell_y)\!\in\!\mathbb{Z}^2 \!:\! \left(\frac{\ell_x \lambda}{L_{R,x}}\right)^2 + \left(\frac{\ell_y \lambda}{L_{R,y}}\right)^2 \!\le 1\right\}\!
\end{align}
at source and receiver, respectively.
We call $n_S = |\mathcal{E}_S|$ and $n_R=|\mathcal{E}_R|$ the cardinalities of the sets $\mathcal{E}_S$ and $\mathcal{E}_R$, respectively. These are given by~\cite{PizzoSPAWC20,PizzoTSP21}
\begin{align}\label{eq:ns}
n_S &=  \Big\lceil{\frac{\pi}{\lambda^2}L_{S,x}L_{S,y}}\Big\rceil + o\left( \frac{L_{S,x}L_{S,y}}{\lambda^2}\right) \\\label{eq:nr}
n_R &= \Big\lceil{\frac{\pi}{\lambda^2}L_{R,x}L_{R,y}}\Big\rceil + o\left( \frac{L_{R,x}L_{R,y}}{\lambda^2}\right)
\end{align}
where $o(\cdot)$ terms can be neglected under Assumption~\ref{Assumption1}. 
With the above definitions at hand, Theorem~\ref{th:series_expansion} is given.

\begin{theorem} \label{th:series_expansion}[Fourier plane-wave series expansion]
For any $s_z$ and $r_z>s_z$, $h(\vect{r},\vect{s})$ within $(r_x,r_y)\in \mathcal{R}$ and $(s_x,s_y)\in \mathcal{S}$ can approximately be described by
\begin{align} \notag
&h(\vect{r},\vect{s}) =  \mathop{\sum}_{(\ell_x,\ell_y)\in\mathcal{E}_R} \mathop{ \sum}_{(m_x,m_y)\in \mathcal{E}_S}  a_R(\ell_x,\ell_y,\vect{r}) H_a(\ell_x,\ell_y,m_x,m_y)   \\& \hspace{.6cm} a_S(m_x,m_y,\vect{s}) \label{Fourier_series}
\end{align}
with random Fourier coefficients 
\begin{equation} \label{Fourier_coeff}
H_a(\ell_x,\ell_y,m_x,m_y)\sim\CN\Big(0,\sigma^2(\ell_x,\ell_y,m_x,m_y)\Big)
\end{equation}
which are statistically independent, circularly-symmetric, complex-Gaussian random variables, each having variance 
\begin{align}\notag
 \sigma^2 &(\ell_x,\ell_y,m_x,m_y) = \\&\frac{1}{(2\pi)^4} \iiiint_{\mathcal{W}_S(m_x,m_y) \times\mathcal{W}_R(\ell_x,\ell_y)} \hspace{-2.5cm}
  S(k_x,k_y,\kappa_x,\kappa_y)  \, dk_xdk_y d\kappa_xd\kappa_y\label{variances}
\end{align}
where the sets $\mathcal{W}_S(m_x,m_y) $ and $\mathcal{W}_R(\ell_x,\ell_y)$ are defined as
\begin{align}\label{S_s}
  \!&\left\{\!\Big[\frac{2\pi m_x}{L_{S,x}},\frac{2\pi (m_x\!+\!1)}{L_{S,x}}\Big] \!\times \!\Big[\frac{2\pi m_y}{L_{S,y}},\frac{2\pi (m_y\!+\!1)}{L_{S,y}}\Big] \!\right\}\!\\\label{S_r}
 \!&\left\{\Big[\frac{2\pi \ell_x}{L_{R,x}},\frac{2\pi (\ell_x+1)}{L_{R,x}}\Big] \!\times\! \Big[\frac{2\pi \ell_y}{L_{R,y}},\frac{2\pi (\ell_y+1)}{L_{R,y}}\Big]\right\}
\end{align}
and are {determined by the $xy$-dimensions of the arrays.}
\end{theorem}

\begin{proof}
The proof {is given in Appendix and} is articulated in two parts.
The first part provides an approximation of~\eqref{Fourier_planewave} over a fixed set of directions. This extends the proof in \cite[App.~IV.A]{PizzoJSAC20} for the receiver only. 
The second part computes the variances of the Fourier coefficients in \eqref{Fourier_coeff} for a generic non-isotropic scenario, by extending \cite[App.~IV.C]{PizzoJSAC20}, valid only for isotropic scattering.
\end{proof}

The above theorem generates a periodic stationary random field that well approximates $h(\vect{r},\vect{s})$ over its fundamental period $(r_x,r_y) \in \mathcal{R}$ and $(s_x,s_y) \in \mathcal{S}$ through the Fourier plane-wave series expansion in \eqref{Fourier_series}.
For the source and receive arrays illustrated in Fig.~\ref{fig:propagation} of $xy$-dimensions ${L_{S,x},L_{S,y}}$ and ${L_{R,x},L_{R,y}}$, the spatial replicas generated by \eqref{Fourier_series} are non-overlapping; hence, the spatial aliasing condition is always satisfied. 
The periodic random field in~\eqref{Fourier_series} is decomposed into the fixed discretized {plane waves} in \eqref{plane_wave_discrete_tx} and \eqref{plane_wave_discrete_rx}, which are
weighted by the random coefficients $H_a(\ell_x,\ell_y,m_x,m_y)$.\footnote{Notice that~\eqref{Fourier_series} is  convergent {in the mean-squared-error sense} as $S(k_x,k_y,\kappa_x,\kappa_y)$ in \eqref{psd_4d} is singularly-integrable \cite{PizzoJSAC20,PizzoIT21}.}
The approximation error in~\eqref{Fourier_series} reduces as $\min(L_{S,x},L_{S,y})/\lambda$ and $\min(L_{R,x},L_{R,y})/\lambda$ become large \cite{PizzoJSAC20}, and vanishes to infinity -- when the two representations in \eqref{Fourier_planewave} and \eqref{Fourier_series} coincide. Hence, it provides an accurate channel description under Assumption~\ref{Assumption1}. To showcase that a good approximation is already achieved for the practical values (i.e., $\min(\cdot,\cdot)/\lambda \ge 10$) envisioned in future high frequency communications and large antenna array technologies, comparisons will be made in Section~\ref{sec:physical_model_variances} and Section~\ref{sec:capacity} with the Clarke's model that is \emph{exact} under isotropic propagation conditions. Similar observations can be found in~\cite{PizzoSPAWC20,PizzoTSP21}.

\subsection{Physical considerations} 

An intuitive physical interpretation of the Fourier plane-wave series expansion in Theorem~\ref{th:series_expansion} is obtained by comparing \eqref{Fourier_series} to its continuous counterpart in \eqref{Fourier_planewave}.
The continuum incident $a_S(\kappa_x,\kappa_y,\vect{s})$ and received $a_R(k_x,k_y,\vect{r})$ {plane waves} are replaced by their discretized versions $a_S(m_x,m_y,\vect{s})$ and $a_R(\ell_x,\ell_y,\vect{r})$, respectively. Physically, this implies that only a finite number of {plane waves} carries the essential channel information available between the two arrays of compact size. 
Consequently, the continuous angular response $H_a(k_x,k_y,\kappa_x,\kappa_y)$ is replaced by the sequence $\{H_a(\ell_x,\ell_y,m_x,m_y)\}$, which, as illustrated in Fig.~\ref{fig:angular_representation}(a), describe the channel coupling between every pair of source $\mathcal{W}_S(m_x,m_y)$ and receive $\mathcal{W}_R(\ell_x,\ell_y)$ angular sets pointed by the above discretized {plane waves}. For this reason, we call $\{H_a(\ell_x,\ell_y,m_x,m_y)\}$ the \emph{coupling coefficients} (e.g., \cite{Miller}). 
By inspection of \eqref{S_s} and \eqref{S_r}, the size of each angular set is inversely proportional to the array size. This is a well known property: spatially larger arrays have higher angular resolution~\cite{Sayeed2002}. Since the angular sets provide a non-overlapping partition of the support $\mathcal{D}(\kappa) \times \mathcal{D}(\kappa)$, it follows that 
\begin{equation} \label{power_variances}
P = \mathop{\sum}_{(\ell_x,\ell_y)\in\mathcal{E}_R} \mathop{ \sum}_{(m_x,m_y)\in \mathcal{E}_S} 
\sigma^2(\ell_x,\ell_y,m_x,m_y).
\end{equation}
Notice that the variances are always bounded since the integrand in \eqref{variances} is singularly-integrable \cite{PizzoIT21,PizzoJSAC20}. 
If $P$ is normalized to unity, then the variances $\sigma^2(\ell_x,\ell_y,m_x,m_y)$ can be regarded as a \emph{discrete angular power distribution} of the channel, which specifies the fraction of power that is transferred from $\mathcal{W}_S(m_x,m_y)$ to $\mathcal{W}_R(\ell_x,\ell_y)$. The strength of the coupling coefficients is not all equal and depends jointly by the array sizes and scattering mechanism. This joint combination determines two key factors: the \emph{number of parallel channels} and the \emph{level of diversity}. The former is specified by the number of source and receive angular sets that are coupled, and provide the DoF of the electromagnetic channel as elaborated in {Section~\ref{sec:KL_expansion}}. The latter is determined by the number of receive angular sets that couple with each {transmit} angular set. Indeed, radiating towards specific angular directions at the source may illuminate several angular directions at the receiver.

As an example, in Fig.~\ref{fig:angular_representation}(b) the coupling coefficients are arranged in matrix form. Three distinct angular sets are activated by the source (orange, blue, and green) whose radiated power is transferred to six angular sets at receiver. The number of parallel channels is three and the level of diversity is three (orange), two (blue), and one (green), respectively.
A possible way to measure the coupling coefficients and their strength from spatial realizations of the channel is described {in Section~\ref{sec:measurement_coupling_coeff}}. An analytical method for their modeling and computation is described next.

\begin{figure}
     \centering
     \begin{subfigure}[b]{\columnwidth}
         \centering
         \includegraphics[width=.95\columnwidth]{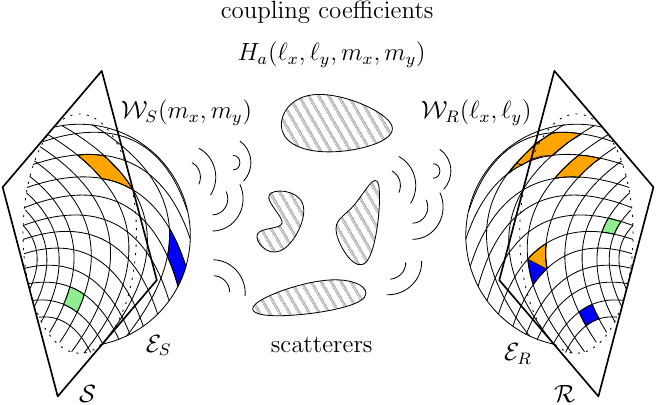}
         \caption{Coupling coefficients}
     \end{subfigure}
     \hfill\vspace{0.2cm}
     \begin{subfigure}[b]{\columnwidth}
         \centering
         \includegraphics[width=.65\columnwidth]{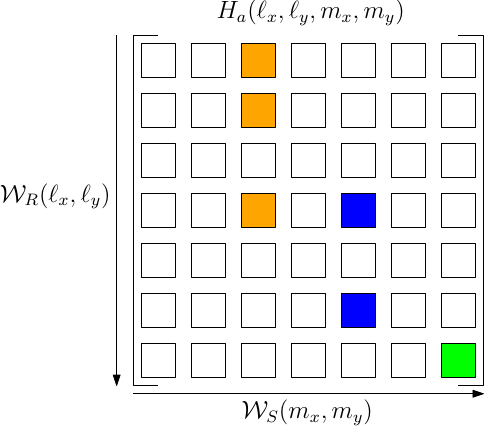}
         \caption{Parallel channels and diversity}
     \end{subfigure}
        \caption{{\small Physical interpretation of the Fourier plane-wave series expansion in Theorem~\ref{th:series_expansion}.}}
        \label{fig:angular_representation}
\end{figure}

\begin{figure}
     \centering
     \begin{subfigure}[b]{\columnwidth}
         \centering
         \includegraphics[width=.95\columnwidth]{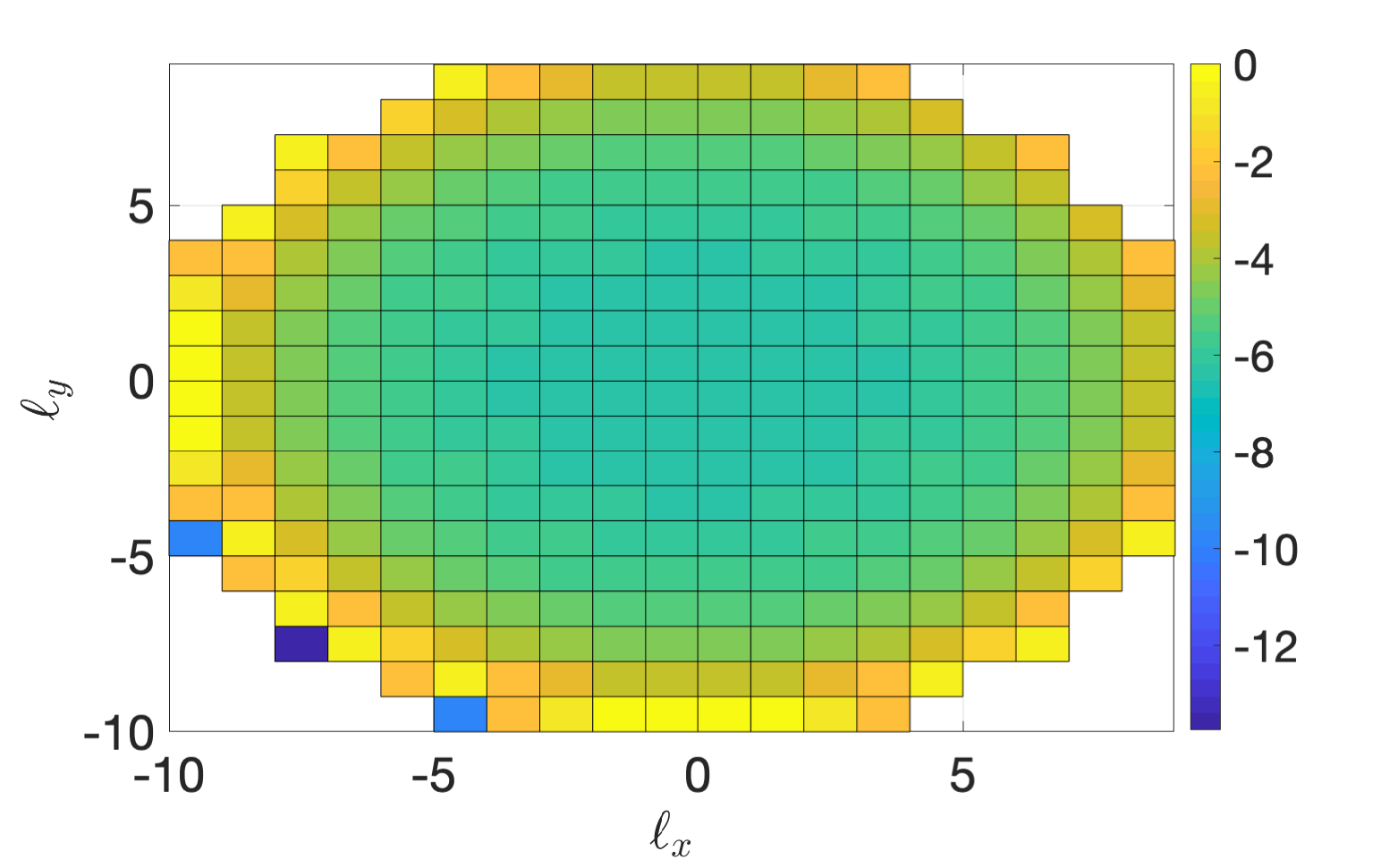}
         \caption{$L/\lambda=10$}
     \end{subfigure}
     \hfill
     \begin{subfigure}[b]{\columnwidth}
         \centering
         \includegraphics[width=.95\columnwidth]{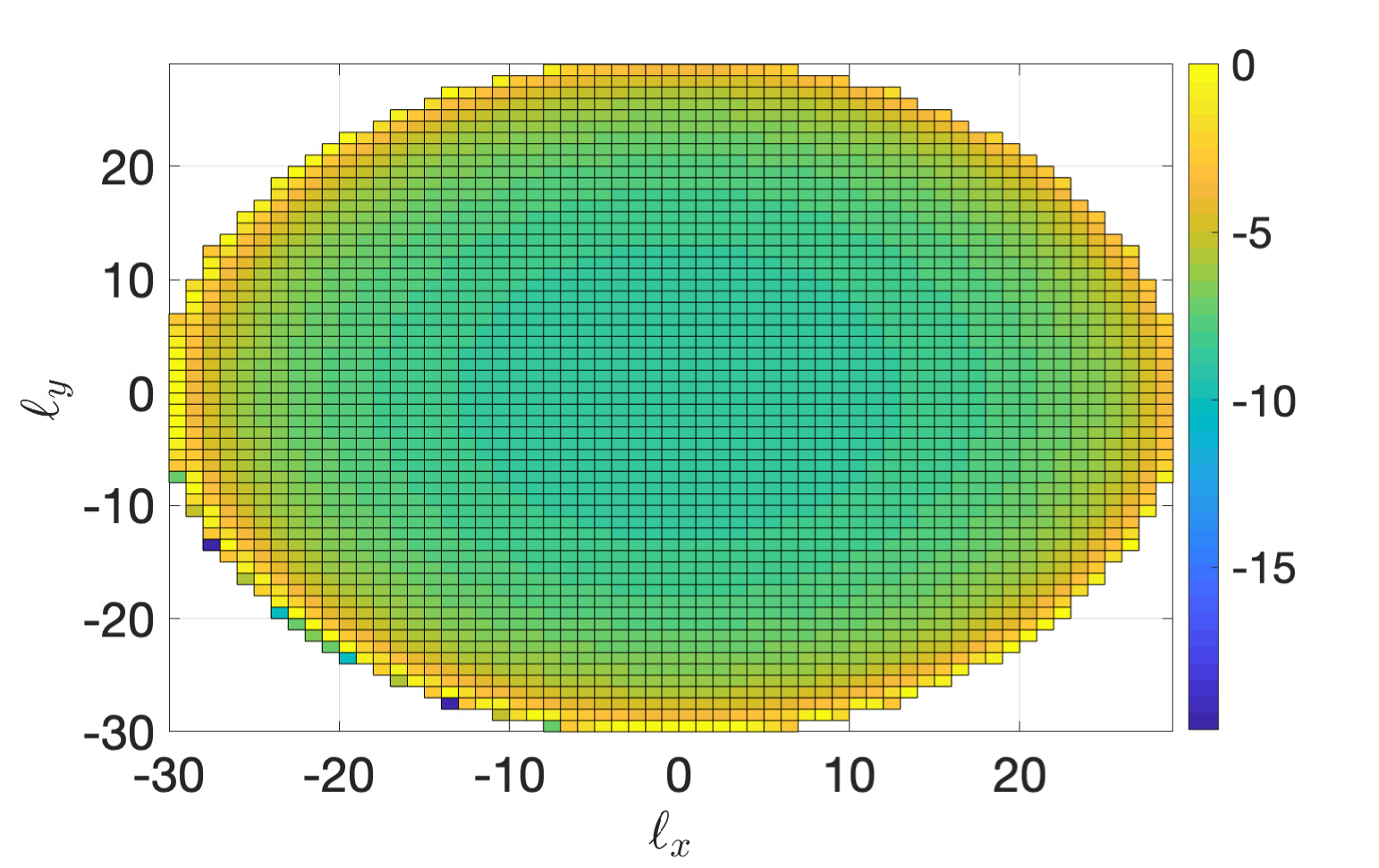}
         \caption{$L/\lambda=30$}
     \end{subfigure}
        \caption{{\small Normalized (to their maximum value) strength of coupling coefficients $N_R \sigma_R^2(\ell_x,\ell_y)$ (in dB) for a squared array of sizes $L/\lambda = \{10,30\}$ with isotropic scattering.}}
        \label{fig:varIso}
\end{figure}

\subsection{Physical modeling of coupling coefficients} \label{sec:physical_model_variances}

In Appendix, after a change of integration variables in~\eqref{variances} from wavenumber $(k_x,k_y,\kappa_x,\kappa_y)$ to spherical coordinates $(\theta_R,\phi_R,\theta_S,\phi_S)$ (i.e., elevation and azimuth angles) we obtain
\begin{align}  \notag
\sigma^2(\ell_x,\ell_y,&m_x,m_y) =\\ & \iiiint_{\Omega_S(m_x,m_y) \times \Omega_R(\ell_x,\ell_y)} 
\hspace{-2.5cm} A^2(\theta_R,\phi_R,\theta_S,\phi_S)   \, d\Omega_S d\Omega_R \label{variance_spherical}
\end{align}
where $\Omega_R(\ell_x,\ell_y)$ and $\Omega_S(m_x,m_y)$ are depicted in Fig.~\ref{fig:integration_subregion} and $d\Omega_R = \sin \theta_R d\theta_R d\phi_R$ and $d\Omega_S = \sin \theta_S d\theta_S d\phi_S$. Clearly, $A^2(\theta_R,\phi_R,\theta_S,\phi_S)$ represents the average angular power transfer in spherical coordinates, which determines the fraction of power transmitted onto $\Omega_S(m_x,m_y)$ and received over $\Omega_R(\ell_x,\ell_y)$. 
This shows that power transfer is thus generally coupled between arrays. 

Practical modeling of $A(\theta_R,\phi_R,\theta_S,\phi_S)$ requires collection of angular measurements of the channel in a prescribed propagation scenario. For analytical purposes, we next illustrate a few simple, but insightful, examples of how to model this function. 
The simplest case is to assume $A^2(\theta_R,\phi_R,\theta_S,\phi_S) = 1$, which corresponds to isotropic propagation.  In this case, the scattering decouples and \eqref{variance_spherical} becomes
\begin{align}  \label{separability_variance}
&\sigma^2(\ell_x,\ell_y,m_x,m_y) =  \sigma_S^2(m_x,m_y) \sigma_R^2(\ell_x,\ell_y)
\end{align}
where $\sigma_S^2(m_x,m_y)$ and $\sigma_R^2(\ell_x,\ell_y)$ account for the power transfer at source and receiver, separately.
Under isotropic scattering, these can be computed in closed-form~\cite[App.~IV.C]{PizzoJSAC20} and physically correspond to solid angles. At the receiver, this is given by 
\begin{align}  \label{solid_angle_rx}
    \sigma_R^2(\ell_x,\ell_y)=|\Omega_R(\ell_x,\ell_y)|  &=  \iint_{\Omega_R(\ell_x,\ell_y) } \hspace{-.5cm} \sin \theta_R d\theta_R d\phi_R
\end{align}
which is uniquely determined by the {$xy$-dimensions of the arrays.}
A direct consequence of~\eqref{separability_variance} is that the fraction of received channel power is the same irrespective from where it is emanated. This tends to result into an optimistic assessment of the number of parallel channels and the level of diversity~\cite[Sec.~3.6]{heath_lozano_2018}.
A simple way to model non-isotropic propagation conditions, while retaining some semblance of the physical reality, is to assume the channel power transfer to be clustered around some $N_{\rm c}\ge1$ modal directions and distributed uniformly within each cluster angular region \cite[Sec.~2]{PoonDoF}. This implies to model the spectral factor as a bounded piecewise constant function over non-overlapped angular sets. In this case, $\sigma^2(\ell_x,\ell_y,m_x,m_y)$ is still decoupled with $ \sigma_R^2(\ell_x,\ell_y)  =  \iint_{\Omega_R(\ell_x,\ell_y) \cap \Theta_R}  \sin \theta_R d\theta_R d\phi_R$
where $\Theta_R$ is the union of all cluster angular regions.
Unlike isotropic propagation, the coupling coefficients are determined by the array spatial resolution and the scattering mechanism jointly; they are non-zero only if $\Omega_S(m_x,m_y)$ and $\Omega_R(\ell_x,\ell_y)$ are (at least partially) subtended by the scattering clusters. A generalized version of the model in~\cite{PoonDoF} is obtained by choosing $A^2(\theta_R,\phi_R,\theta_S,\phi_S) = A_R^2(\theta_R,\phi_R) A_S^2(\theta_S,\phi_S)$ with $A_S(\theta_S,\phi_S)$ and $A_R(\theta_R,\phi_R)$ varying arbitrarily. This yields
\begin{align} \label{variances_rx}
    \sigma_R^2(\ell_x,\ell_y)  &=  \iint_{\Omega_R(\ell_x,\ell_y)} \hspace{-.5cm}A_R^2(\theta_R,\phi_R)  \sin \theta_R d\theta_R d\phi_R.
\end{align}
A good trade-off between tractability and accuracy is offered by the mixture of 3D von Mises-Fisher (vMF) family of angular density functions~\cite{PizzoIT21}. At receiver, e.g., it yields
\begin{equation} \label{VMF_spectral_factor}
A_R^2(\theta_R,\phi_R) = \sum_{i=1}^{N_{\rm c}} w_i \, p_{R,i}(\theta_R,\phi_R)
\end{equation}
with positive weights such that $\sum_i w_i=1$ and $p_{R,i}(\theta_R,\phi_R) = c(\alpha_i) e^{\alpha_i (\sin\theta \sin\mu_{\theta,i} \cos(\phi-\mu_{\phi,i}) + \cos\theta \cos\mu_{\theta,i})}$. Here, $c(\alpha_i) = \alpha_i/(4\pi \sinh \alpha_i)$ is a normalization constant, $\{\mu_{\theta,i},\mu_{\phi,i}\}$ represent the elevation and azimuth angles of the modal direction and $\alpha_i$ is the so-called concentration parameter {for each cluster $i=1,\ldots,N_{\rm c}$.} The formers specify the propagation direction around which the power is concentrated while the latter determines the concentration of angular power, that is, as $\alpha$ increases the density becomes more concentrated around its modal direction. This is directly related to the circular variance $\nu^2 \in[0,1]$ of each cluster \cite[Eqs.~(77)--(78)]{PizzoIT21}, which, in turn, determines the angular spread of the channel $\sigma_a \in[0^\circ,360^\circ]$. Notice that the isotropic case is obtained by setting $N_{\rm c}=1$ and $\alpha_1=0$.

\begin{figure}
     \centering
     \begin{subfigure}[b]{\columnwidth}
         \centering
         \includegraphics[width=.95\columnwidth]{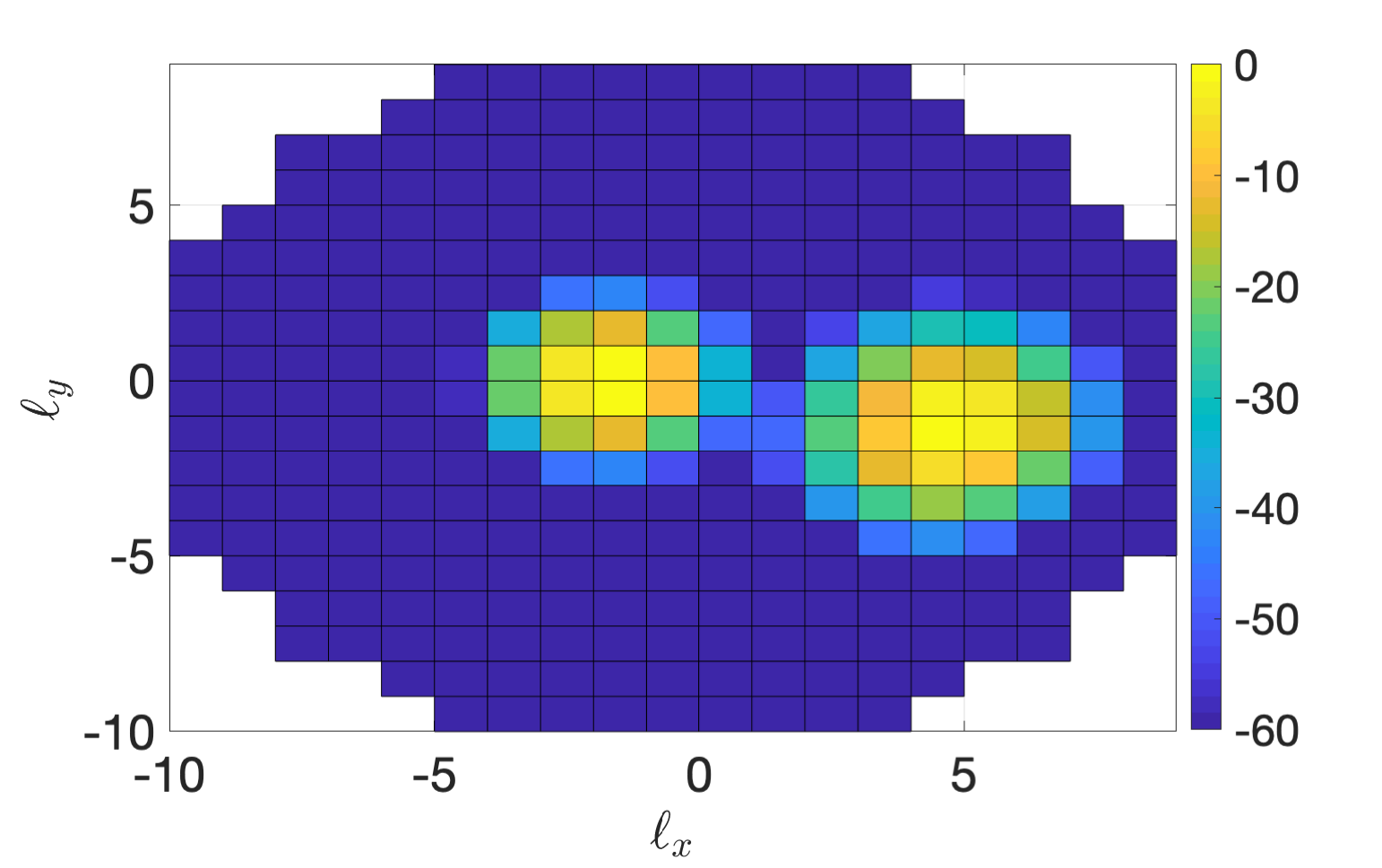}
         \caption{$L/\lambda=10$}
     \end{subfigure}
     \hfill
     \begin{subfigure}[b]{\columnwidth}
         \centering
         \includegraphics[width=.95\columnwidth]{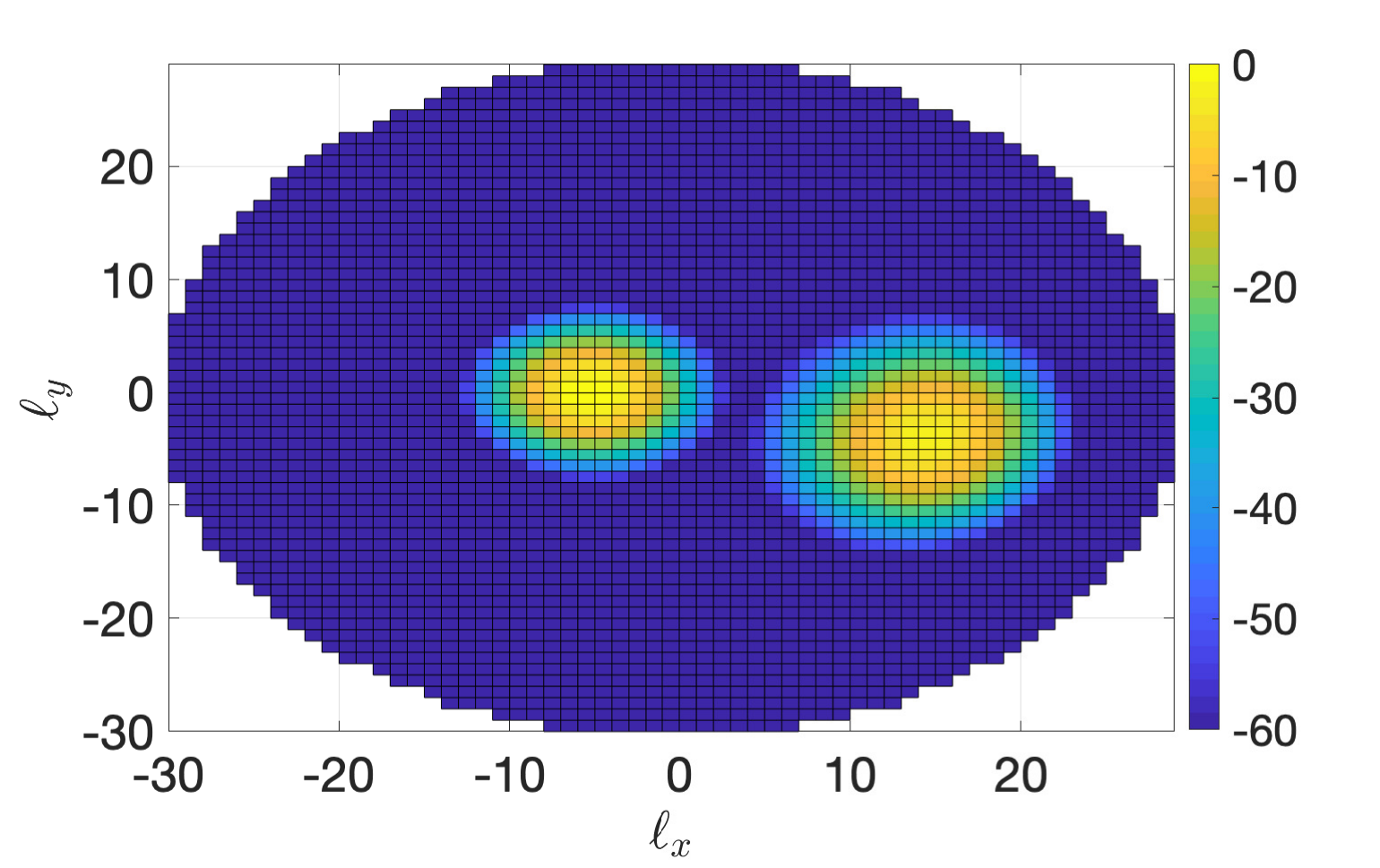}
         \caption{$L/\lambda=30$}
     \end{subfigure}
        \caption{{\small Normalized (to their maximum value) strength of coupling coefficients $N_R \sigma_R^2(\ell_x,\ell_y)$ (dB) for $L/\lambda = \{10,30\}$ with non-isotropic scattering. 
The mixture of 3D vMF angular density function is used with $N_{\rm c}=2$, $\mu_\theta = \{30^\circ,10^\circ\}$, $\mu_\phi= \{345^\circ,180^\circ\}$ and $\alpha$ such that $\nu^2 = \{.01,.005\}$, which roughly corresponds to an angular spread of $\sigma_a = \{15^\circ ,11^\circ\}$.}}
\label{fig:varNonIso}
\end{figure}

\begin{figure*}
\begin{align} \tag{36}\label{Fourier_series_karhunen}
h(\vect{r},\vect{s}) = \mathop{\sum}_{(\ell_x,\ell_y)\in\mathcal{E}_R} \mathop{ \sum}_{(m_x,m_y)\in \mathcal{E}_S}  \widetilde H(\ell_x,\ell_y,m_x,m_y;r_z,s_z)  \phi_R(\ell_x,\ell_y,r_x,r_y)  \phi_S^*(m_x,m_y,s_x,s_y)
\end{align}
\hrule
\begin{align}\label{correlation_karhunen}
c(r_x,r_y,s_x,s_y) = \mathop{\sum}_{(\ell_x,\ell_y)\in\mathcal{E}_R} \mathop{ \sum}_{(m_x,m_y)\in \mathcal{E}_S}   \sigma^2(\ell_x,\ell_y,m_x,m_y)  \phi_R(\ell_x,\ell_y,r_x,r_y)  \phi_S^*(m_x,m_y,s_x,s_y)\tag{38}
\end{align}
\hrule
{\small
\begin{align}\notag
 S&(k_x,k_y,\kappa_x,\kappa_y) = \mathop{\sum}_{(\ell_x,\ell_y)\in\mathcal{E}_R} \mathop{ \sum}_{(m_x,m_y)\in \mathcal{E}_S} \!\!\! \sigma^2(\ell_x,\ell_y,m_x,m_y)  \,
\delta\left(k_x - \frac{2\pi \ell_x}{L_{R,x}}\right) \delta\left(k_y - \frac{2\pi \ell_y}{L_{R,y}}\right) 
\delta\left(\kappa_x - \frac{2\pi m_x}{L_{S,x}}\right) \delta\left(\kappa_y - \frac{2\pi m_y}{L_{S,y}}\right)\label{psd_karhunen} \tag{39}
\end{align}}
\hrule
\end{figure*}

Assume symmetric spectral factors, i.e., $A_R(\theta_R,\phi_R) =A_S(\theta_S,\phi_S)$, and focus on the receive only. Assume also ${L_{R,x}=L_{R,y} =L}$ with ${L/\lambda = \{10,30\}}$.
Based on~\eqref{eq:nr}, there are essentially ${n_R \approx \lceil \pi (L/\lambda)^2\rceil= \{315, 2828\}}$ angular sets.
Under isotropic scattering, the normalized variances $N_R \sigma_R^2(\ell_x,\ell_y)$ in~\eqref{variances_rx} are plotted in Fig.~\ref{fig:varIso}.
As expected, the size of each angular set reduces as $L$ increases due to a higher angular resolution of the array.
The non-zero coupling coefficients are exactly {$n_R = \{344, 2928\}$} and only distributed within the lattice ellipse $\mathcal{E}_R$ in \eqref{epsilon_r} according to a bowl-shaped behavior. Thus, they are never all equal even under isotropic scattering. This observation will be used in Section~\ref{sec:correlation_matrix} to conclude that an electromagnetic random MIMO channel must \emph{necessarily} exhibits spatial correlation. 

The non-isotropic propagation is considered in Fig.~\ref{fig:varNonIso}. Here, we have $N_{\rm c}=2$ with $\mu_\theta = \{30^\circ,10^\circ\}$, $\mu_\phi= \{345^\circ,180^\circ\}$ and $\alpha$ such that $\nu^2 = \{.01,.005\}$, which roughly corresponds to $\sigma_a = \{15^\circ ,11^\circ\}$. The coupling coefficients are still non-zero within the lattice ellipse $\mathcal{E}_R$ in \eqref{epsilon_r} but, unlike Fig.~\ref{fig:varIso}, they achieve higher values around the modal directions. Compared to the isotropic case, a reduced number, say ${n_R^\prime < n_R}$, of coupling coefficients is significant.
Inspired by the three-sigma rule for Gaussian distributions, $n_R^\prime$ may be computed as the number of angular sets that are sufficient to capture the $99.7\%$ of the channel power in~\eqref{power_variances}. This yields ${n_R^\prime = 21+14=35}$ and ${n_R^\prime = 145+84=229}$ for Fig.~\ref{fig:varNonIso}(a) and Fig.~\ref{fig:varNonIso}(b), respectively.

\subsection{Connection to Karhunen-Loeve expansion}

Let $h(t)$ be a band-limited stationary random process of bandwidth $B$ that is observed over a time interval $[0,T]$. In the regime $BT \gg 1$, the eigenfunctions of its Karhunen-Loeve expansion approach complex harmonics oscillating at an integer multiple of the fundamental frequency $1/T$. The eigenvalues' power are obtained by sampling the power spectral density of the process at these frequencies \cite[Sec.~3.4]{VanTreesBook}.
Next, we show how this fundamental result applies to spatially-stationary electromagnetic fields, which we recall to be band-limited in the spatial-frequency domain with maximum circularly bandwidth $\pi\kappa^2$~\cite{PizzoJSAC20,PizzoIT21,PizzoTSP21}.

Consider the two 2D spatial-frequency Fourier harmonics
\begin{align} \label{harmonic_discrete_tx}
\phi_S(m_x,m_y,s_x,s_y)  &= e^{\imagunit \left(\frac{2\pi}{L_{S,x}} m_x s_x + \frac{2 \pi}{L_{S,y}} m_ys_y\right)}  \\
\phi_R(\ell_x,\ell_y,r_x,r_y)  &= e^{\imagunit \left(\frac{2\pi}{L_{R,x}} \ell_xr_x + \frac{2 \pi}{L_{R,y}} \ell_yr_y\right)}
\end{align}
with fundamental periods $(r_x,r_y) \in \mathcal{R}$ and $(s_x,s_y) \in \mathcal{S}$.
For any fixed pair $(r_z,s_z)$, the discretized {plane waves} in~\eqref{Fourier_series} correspond to two phase-shifted versions of the two 2D spatial-frequency Fourier harmonics, i.e.,
\begin{align} \label{plane_wave_discrete_tx_Fourier}
a_S(m_x,m_y,\vect{s}) & = \phi_S^*(m_x,m_y,s_x,s_y) e^{-\imagunit  \gamma_S(m_x,m_y) s_z} \\ \label{plane_wave_discrete_rx_Fourier}
a_R(\ell_x,\ell_y,\vect{r}) & = \phi_R(\ell_x,\ell_y,r_x,r_y) e^{\imagunit  \gamma_R(\ell_x,\ell_y) r_z}.
\end{align} 
Hence, we can rewrite~\eqref{Fourier_series} as~\eqref{Fourier_series_karhunen}
where we have defined
\setcounter{equation}{36}
\begin{align}\notag
  \widetilde H&(\ell_x,\ell_y,m_x,m_y;r_z,s_z) \\& = H_a(\ell_x,\ell_y,m_x,m_y)    e^{-\imagunit \gamma_S(m_x,m_y) s_z} e^{\imagunit \gamma_R(\ell_x,\ell_y) r_z}.\label{spectral_response}
 \end{align}
Notice that $\widetilde H(\ell_x,\ell_y,m_x,m_y;r_z,s_z)$ and $H_a(\ell_x,\ell_y,m_x,m_y)$ are statistically equivalent due to~\eqref{Fourier_coeff}. Hence, we can remove the dependance on $r_z$ and $s_z$ on the spatial correlation function $c(\vect{r},\vect{s}) = \Ex\{h(\vect{r}+\vect{r}^\prime,\vect{s}+\vect{s}^\prime) h^*(\vect{r}^\prime,\vect{s}^\prime)\}$ of $h(\vect{r},\vect{s})$.
By using \eqref{Fourier_series_karhunen} and \eqref{spectral_response}, it can be approximated as in~\eqref{correlation_karhunen}.

The closed-form expression~\eqref{correlation_karhunen} can be regarded as the \emph{asymptotic} Hilbert-Schmidt decomposition \cite[Sec.~3.4]{FranceschettiBook} of the self-adjoint correlation kernel function $c(\vect{r},\vect{s})$; that is, $\{\phi_S(m_x,m_y, \vect{s})\}$ and $\{\phi_R(\ell_x,\ell_y, \vect{r})\}$ are the complete (non-normalized) orthonormal basis sets of eigenfunctions, and $\{\sigma^2(\ell_x,\ell_y,m_x,m_y)\}$ is the sequence of non-negative real-valued eigenvalues.
In analogy with the time-domain stationary case \cite{VanTreesBook}, the expansion of $h(\vect{r},\vect{s})$ over the above basis sets of eigenfunctions yields~\eqref{Fourier_series_karhunen}, which is the \emph{asymptotic} Karhunen-Loeve expansion~\cite[Sec. 6.4]{FranceschettiBook} of a spatially-stationary electromagnetic random field. As for time-domain processes, the asymptotic regime is achieved under Assumption~\ref{Assumption1}.
Unlike the time-domain case, the spatial case exhibits a lower-dimensionality since the {six-dimensional} power spectral density of $h(\vect{r},\vect{s})$ is impulsive and defined on a double sphere of radius $\kappa$ \cite{PizzoJSAC20,PizzoIT21}. Another key difference is that we do not sample at integer multiples of the fundamental spatial frequencies, but rather integrate \eqref{psd_4d} over a neighborhood of these frequencies. As shown in Appendix, this is because \eqref{psd_4d} is singularly-integrable. Notice that the standard sampling of the power spectral density in~\eqref{psd_4d} at multiple of the fundamental spatial frequencies would have generated a divergent series expansion.
By applying a 4D spatial Fourier transform to~\eqref{correlation_karhunen} we obtain a power spectral density of the form in~\eqref{psd_karhunen}, which is impulsive due to the periodic nature of $c(\vect{r},\vect{s})$.  Nevertheless, it can by shown that~\eqref{psd_karhunen} tends to~\eqref{psd_4d} asymptotically as $\min(L_{S,x},L_{S,y})/\lambda\to \infty$ and $\min(L_{R,x},L_{R,y})/\lambda\to \infty$.

\section{Stochastic Electromagnetic MIMO Channel Model}\label{sec:MIMO_channel_model}
\setcounter{equation}{39}
From~\eqref{eq:MIMO_entries}, the MIMO channel matrix is approximated by sampling the series representation in~\eqref{Fourier_series} as
\begin{align} \notag
\vect{H}  = &   \sqrt{N_R N_S}  \mathop{\sum}_{(\ell_x,\ell_y)\in\mathcal{E}_R} \mathop{ \sum}_{(m_x,m_y)\in \mathcal{E}_S} \\& \times H_a(\ell_x,\ell_y,m_x,m_y)  \vect{a}_R(\ell_x,\ell_y)  \vect{a}_S^{\Htran}(m_x,m_y)\!\!\label{Fourier_series_discrete} 
\end{align}
where $\vect{a}_R(\ell_x,\ell_y)$ and  $\vect{a}_S(m_x,m_y)$ represents the (normalized) discrete-space source and receive array responses with entries
\begin{align} \label{a_s_vec}
\left[\vect{a}_S(m_x,m_y)\right]_j & = \frac{1}{\sqrt{N_S}} a_S(m_x,m_y,{s}_{x_j},{s}_{y_j},s_z) \\ \label{a_r_vec}
\left[\vect{a}_R(\ell_x,\ell_y)\right]_i &= \frac{1}{\sqrt{N_R}} a_R(\ell_x,\ell_y,{r}_{x_j},{r}_{y_j},r_z).
\end{align}
Different array geometries and antenna spacings may have a marked effect on $\vect{H}$ and its statistics. To guarantee that no information is lost by sampling $h(\vect{r},\vect{s})$, the Nyquist condition in the spatial domain must be satisfied. From Corollary~\ref{th:bandlimited},~$h(\vect{r},\vect{s})$ is a circularly-bandlimited channel with maximum bandwidth $\pi\kappa^2$, for any scattering environment. For a uniform spatial sampling of $h(\vect{r},\vect{s})$, the Nyquist condition is met when antenna separation is at most half-wavelength (e.g.,\cite[Sec. V]{PizzoJSAC20}, \cite{PizzoTSP21}).
In other words, no information is lost when the transmit and receive arrays are equipped with
\begin{align}\label{eq:numofantN_S}
N_S &\ge \frac{4 L_{S,x} L_{S_y}}{\lambda^2} \ge n_S \\
N_R &\ge \frac{4 L_{R,x} L_{R_y}}{\lambda^2} \ge n_R\label{eq:numofantN_r}
\end{align}
antenna elements under Assumption~\ref{Assumption1}. {Both conditions are assumed to be satisfied in the remainder. Also, we assume that antenna spacing is uniform, leading to a uniform spatial sampling of~\eqref{Fourier_series}}.
Notice that~\eqref{eq:numofantN_S}-\eqref{eq:numofantN_r} ensure that the number of DoF created by the the source and resolved by the receiver is not less than the maximum DoF that may possibly be generated by the scattering, i.e., $\min(n_R,n_S)$. In practice, this number is limited by the richness of the scattering.

\begin{remark}
We stress that sampling at Nyquist's rate (or higher) ensures no loss of channel information. This allows to fully exploit the propagation characteristics offered by an electromagnetic channel and to design a system that ultimately exploits all its DoF, which is exactly the scope of a Holographic MIMO system.
\end{remark}

There is a 2D counterpart to the 3D theory presented in this paper, where wave propagation takes place on a 2D plane rather than a 3D space; see \cite{PizzoJSAC20,PizzoIT21}. In this case, a similar representation of $\vect{H}$ is obtained, which is reminiscent of the virtual channel representation in~\cite{Sayeed2002,Veeravalli}.
Both provide an angular decomposition of a 2D MIMO channel over a fixed sets of directions that are specified by statistically-independent random coefficients. 
However, the Fourier plane-wave model differs from the virtual channel representation in several aspects: \emph{i}) it is derived from the physics principles of wave propagation and it is thus valid also in the near-field propagation region; \emph{ii}) it models planar and volumetric arrays of arbitrary geometry operating in a 3D propagation environment; \emph{iii}) it reveals the lower-dimensionality of the angular description for electromagnetic channels, i.e., only $n_R n_S$ (rather than $N_R N_S$) coupling coefficients contain the essential information; \emph{iv}) it supports the physical model with a statistical analysis that is built upon a closed-form expression of the power spectral density of an electromagnetic random channel.

\subsection{Karhunen-Loeve expansion} \label{sec:KL_expansion}

Call ${\boldsymbol{\phi}_S(m_x,m_y)\in\mathbb{C}^{N_S}}$ the vector with entries ${\left[\boldsymbol{\phi}_S(m_x,m_y)\right]_j = \frac{1}{\sqrt{N_S}}{\phi}_S(m_x,m_y,{s}_{x_j},{s}_{y_j})}$
for ${j=1,\ldots,N_S}$. Similarly, $ \boldsymbol{\phi}_R(\ell_x,\ell_y) \in\mathbb{C}^{N_R}$ has entries $\left[\boldsymbol{\phi}_R(\ell_x,\ell_y)\right]_i = \frac{1}{\sqrt{N_R}}{\phi}_R(\ell_x,\ell_y,{r}_{x_j},{r}_{y_j})$
for $i=1,\ldots,N_R$. Hence,~\eqref{Fourier_series_discrete} can equivalently be rewritten as
\begin{align} \notag
\vect{H}  =  & \sqrt{N_R N_S}  \mathop{\sum}_{(\ell_x,\ell_y)\in\mathcal{E}_R} \mathop{ \sum}_{(m_x,m_y)\in \mathcal{E}_S} & \widetilde H(\ell_x,\ell_y,m_x,m_y;r_z,s_z)  \\& \hspace{.1cm} \times \boldsymbol{\phi}_R(\ell_x,\ell_y)  \boldsymbol{\phi}_S^{\Htran}(m_x,m_y) \label{Fourier_series_discrete_KL} 
\end{align}
where $\widetilde H(\ell_x,\ell_y,m_x,m_y;r_z,s_z)$ are given by~\eqref{spectral_response}. 
{With uniform sampling}, $\{\boldsymbol{\phi}_S(m_x,m_y)\}$ and $\{\boldsymbol{\phi}_R(\ell_x,\ell_y)\}$ constitute a set of orthonormal discrete basis functions. Hence,~\eqref{Fourier_series_discrete_KL} can be regarded as the Karhunen-Loeve expansion of the electromagnetic MIMO channel~\cite{Tulino2005}.
Notice that these orthonormal basis are fixed; they do not depend on the statistics of $\vect{H}$. The strength of each coupling coefficient specifies the average amount of energy transmitted by the $(m_x,m_y)$th source basis function that couples with the $(\ell_x,\ell_y)$th receive basis function. The average channel power is thus $\Ex\left\{\tr(\vect{H}^{\Htran}\vect{H})\right\} = P$ in~\eqref{power_variances}. 

Denote $\boldsymbol{\Phi}_S$ and $\boldsymbol{\Phi}_R$ the deterministic matrices collecting the $n_S$ and $n_R$ vectors $\{\boldsymbol{\phi}_S(m_x,m_y)\}$ and $\{\boldsymbol{\phi}_R(\ell_x,\ell_y)\}$, respectively. These are semi-unitary matrices, i.e., $\boldsymbol{\Phi}_S^{\Htran}\boldsymbol{\Phi}_S = \vect{I}_{n_S}$ and $\boldsymbol{\Phi}_R^{\Htran} \boldsymbol{\Phi}_R= \vect{I}_{n_R}$. Precisely, they are obtained by collecting columns of two 2D inverse discrete Fourier transform (IDFT) matrices.  
Let $\boldsymbol{\gamma}_S$ and $\boldsymbol{\gamma}_R$ be the column vectors containing the $n_S$ and $n_R$ coefficients $\gamma_S(m_x,m_y)$ and $\gamma_R(\ell_x,\ell_y)$. The following lemma is thus obtained.

\begin{lemma} \label{th:channel_matrix_lemma}
For any $s_z$ and $r_z > s_z$, the MIMO channel matrix  can approximately be described by
\begin{equation}\label{channel_matrix}
\vect{H} = \boldsymbol{\Phi}_R \widetilde{\vect{H}}  \boldsymbol{\Phi}_S^{\Htran} 
\end{equation}
where $\widetilde{\vect{H}}  = e^{\imagunit\boldsymbol{\Gamma}_R} \vect{H}_a e^{-\imagunit\boldsymbol{\Gamma}_S}$.
Here, $e^{\imagunit\boldsymbol{\Gamma}_R}$ and $e^{-\imagunit\boldsymbol{\Gamma}_S}$ are diagonal matrices with $\boldsymbol{\Gamma}_R = \diag(\boldsymbol{\gamma}_R) r_z $ and $\boldsymbol{\Gamma}_S = \diag( \boldsymbol{\gamma}_S) s_z$,
and $\vect{H}_a\in\mathbb{C}^{n_R\times n_S}$ is the angular random matrix obtained as
\begin{equation}  \label{angular_coefficients}
\vect{H}_a = \boldsymbol{\Sigma} \odot \vect{W}
\end{equation}
where $\boldsymbol{\Sigma} \in \Real_+^{n_R \times n_S}$ collects the $n_Rn_S$ scaled standard deviations $\{\sqrt{N_SN_R} \sigma(\ell_x,\ell_y,m_x,m_y)\}$ and $\vect{W} \in \Complex^{n_R \times n_S}$ is a matrix with i.i.d. circularly-symmetric, complex-Gaussian random entries. 
\end{lemma}
Since $\boldsymbol{\Phi}_S$ and $\boldsymbol{\Phi}_R$ are semi-unitary matrices, and $\widetilde{\vect{H}} $ is statistically equivalent to $\vect{H}_a$, the above lemma shows that the angular matrix $\vect{H}_a$ is \emph{semi-unitarily equivalent} to $\vect{H}$, in the sense that the most significant $\min(n_R,n_S)$ singular values of the two matrices are identical. This property is reminiscent of the unitary-independent-unitary (UIU) MIMO channel model~\cite[Eq.~(1)]{Tulino2005}, where the channel is decomposed in terms of $N_R-$ and $N_S$-dimensional unitary matrices.
However, as $n_S \ll N_S$ and $n_R \ll N_R$, we have that channel models in the UIU form represent a highly redundant description of the electromagnetic channel. In other words, the provided Fourier plane-wave model yields a physics-based \emph{low-rank approximation} of $\vect{H}$ with respect to fixed spatial basis matrices that are independently defined by the array geometry at each end. 
The low-rank property of electromagnetic channels is stronger under non-isotropic propagation conditions when only a subset of coupling coefficients is significant, as discussed in Section~\ref{sec:physical_model_variances}. 
Since the coupling coefficients are independent, the rows and columns of $\vect{H}_a$ are linearly independent with probability $1$ except for those that are identically zero~\cite{Tulino2005}. Particularly, if we denote with $n_S^\prime \le n_S$ and $n_R^\prime \le n_R$ the number of rows and columns of $\vect{H}_a$ that are not identically zero, the number of singular values is
\begin{equation} \label{DoF}
{\rm rank}(\vect{H}_a) = \min(n_S^\prime,n_R^\prime) \le \min(n_S,n_R)
\end{equation}
which corresponds to the DoF of the channel.
As observed in~\cite{PoonDoF}, \cite{PizzoSPAWC20}, the DoF does not scale with the number of antenna elements but depends on the scattering environment and array sizes jointly.
Clearly, the number of DoF equals the upper bound in~\eqref{DoF} when non-zero power is received from all angular sets; that is, under isotropic scattering. In this case, the DoF per m$^2$ are roughly given by $\lceil\pi/\lambda^2\rceil$~\cite{PizzoSPAWC20}, which can be rewritten as {the Landau's formula $\lceil |\mathcal{D}(\kappa)|/(2\pi)^2\rceil$\cite{PizzoTSP21}.} In general, the wavenumber support $\mathcal{K}\subseteq \mathcal{D}(\kappa)$ is limited by the richness of the non-isotropic scattering \cite{PizzoTSP21}.
A rough estimate of its measure {$|\mathcal{K}|$} is obtained by computing the area covered by all rectangular sets in \eqref{S_s} and \eqref{S_r} wherein the spectral factor is non-zero. As expected, this corresponds to the number of non-zero coupling coefficients per m$^2$.

{Notice that whole randomness of $\vect{H}$ is fully embedded into $\vect{H}_a$ whereas the deterministic matrices $\boldsymbol{\Phi}_S^{\Htran}  e^{-\imagunit\boldsymbol{\Gamma}_S}$ and $\boldsymbol{\Phi}_R e^{\imagunit\boldsymbol{\Gamma}_R}$ change the domain of representation from angular to spatial at source and receiver, respectively.
Particularly, the two matrices $e^{\imagunit \boldsymbol{\Gamma}_R}$ and $e^{-\boldsymbol{\imagunit\Gamma}_s}$ are known in physics as \emph{migration filters} and are fully determined by the array geometry and wavelength {\cite{PizzoTSP21}}.
They contain the entire effect due to wave propagation along the $z$-axis. Note that this effect is fully deterministic and known a priori; hence, there is no channel information that can be captured along the $z$-dimension. 
As observed in \cite{Franceschetti}, the world -- at both source and receiver -- has only an apparent 3D informational structure, which is subject to a 2D representation. A key consequence of this observation is that a 3D volumetric array offers no extra DoF over a 2D planar array (e.g.,~\cite{Miller,PizzoSPAWC20}).

The MIMO channel model~\eqref{channel_matrix} does not have a Kronecker structure. However, this naturally arises if the separability propagation condition \eqref{separability_variance} is imposed.
\begin{corollary}[Separable model] \label{th:separable_model}
If the scattering is separable, \eqref{angular_coefficients} becomes
 \begin{equation} \label{matrix_angular_coefficients_kronecker}
\vect{H}_a  =  \diag(\boldsymbol{\sigma}_R) \vect{W} \diag(\boldsymbol{\sigma}_S)
\end{equation}
where ${\boldsymbol{\sigma}_R \in \Real_+^{n_R}}$ and ${\boldsymbol{\sigma}_S \in \Real_+^{n_S}}$ collect $\{\sqrt{N_R} \sigma_R(\ell_x,\ell_y)\}$ and $\{\sqrt{N_S} \sigma_S(m_x,m_y)\}$, respectively. The use of~\eqref{matrix_angular_coefficients_kronecker} into \eqref{channel_matrix} yields the Kronecker channel matrix
\begin{equation}\label{channel_matrix_kronecker}
\vect{H} = \left(\boldsymbol{\Phi}_R \diag(\boldsymbol{\sigma}_R) e^{\imagunit\boldsymbol{\Gamma}_R}  \right) \vect{W} \left( e^{-\imagunit\boldsymbol{\Gamma}_S} \diag(\boldsymbol{\sigma}_S)   \boldsymbol{\Phi}_S^{\Htran} \right).
\end{equation}
\end{corollary}
\begin{proof}
Under the separability condition~\eqref{separability_variance}, $\boldsymbol{\Sigma}$ in \eqref{angular_coefficients} becomes $\boldsymbol{\Sigma} = \boldsymbol{\sigma}_R \boldsymbol{\sigma}_S^{\Ttran} = \boldsymbol{\sigma}_R \vect{1}^{\Ttran}_{n_S} \odot  \vect{1}_{n_R} \boldsymbol{\sigma}_S^{\Ttran}$ and thus $\vect{H}_a  = \left(\boldsymbol{\sigma}_R \vect{1}^{\Ttran}_{n_S}\right) \odot 
\left(\vect{1}_{n_R} \boldsymbol{\sigma}_S^{\Ttran}\right)  \odot \vect{W}$. By applying the matrix identity $\left(\vect{a} \vect{b}^{\Ttran}\right) \odot  \vect{X} = \diag(\vect{a}) \vect{X} \diag(\vect{b})$ twice, we obtain~\eqref{matrix_angular_coefficients_kronecker}.
\end{proof}
A similar model was proposed in~\cite[Eq.~(14) and Eq.~(27d)]{PoonCapacity} for a MIMO channel with linear arrays in the far-field. The model in~\eqref{matrix_angular_coefficients_kronecker} is a generalization for planar arrays {valid also} in the near-field. 

\subsection{Channel Statistics}\label{sec:correlation_matrix}
Lemma~\ref{th:channel_matrix_lemma} gives rise to correlated Rayleigh fading (e.g.,~\cite{LucaBook,heath_lozano_2018}) where {${\bf R}$ in~\eqref{full_correlation_matrix} describes the joint correlation properties of both link ends.}
A detailed description of ${\bf R}$ is thus in order. 
We begin by noticing that $\widetilde{\vect{H}}$ and ${\vect{H}}_a$ are statistically equivalent to each other, i.e., the channel statistics are invariant to any translation of the two arrays along the $z$-axis. In particular, since $e^{\imagunit\boldsymbol{\Gamma}_R}$ and $e^{-\imagunit\boldsymbol{\Gamma}_S}$ are diagonal, by using the matrix identity $ \diag(\vect{a}) (\vect{X} \odot \vect{Y}) \diag(\vect{b}) = \vect{X} \odot (\diag(\vect{a}) \vect{Y} \diag(\vect{b}))$
we may write $\widetilde{\vect{H}}$ as $\widetilde{\vect{H}}  = \boldsymbol{\Sigma} \odot (e^{\imagunit\boldsymbol{\Gamma}_R}{\vect{W}} e^{-\imagunit\boldsymbol{\Gamma}_S})$. Hence, the equivalence follows from~\eqref{angular_coefficients} and the statistical equivalence between $e^{\imagunit\boldsymbol{\Gamma}_R}{\vect{W}} e^{-\imagunit\boldsymbol{\Gamma}_S}$ and ${\vect{W}}$. The correlation matrix is obtained from~\eqref{channel_matrix} as (e.g.,\cite{Weichselberger2006})
\begin{equation}\label{correlation_matrix}
{\bf R}= {\bf U} \boldsymbol{\Lambda} {\bf U}^{\Htran}
\end{equation}
where ${\bf U} =  \boldsymbol{\Phi}_S \kron \boldsymbol{\Phi}_R \in \mathbb{C}^{N_R N_S \times n_R n_S}
$
is semi-unitary\footnote{The Kronecker product of semi-unitary matrices is semi-unitary.} (i.e., ${\bf U}^{\Htran} {\bf U} = \vect{I}_{n_S n_R}$) and 
\begin{equation} \label{eigenvalues_channel}
\boldsymbol{\Lambda} = \diag({\rm vec}(\boldsymbol{\Sigma} \odot \boldsymbol{\Sigma})).
\end{equation}
The expression~\eqref{correlation_matrix} provides the eigendecomposition of ${\bf R}$ with $(\boldsymbol{\phi}_S(m_x,m_y) \kron \boldsymbol{\phi}_R(\ell_x,\ell_y))$ being the eigenvectors and $N_SN_R \sigma^2(\ell_x,\ell_y,m_x,m_y)$ being the eigenvalues. As mentioned before, the eigenvectors are fixed, not derived from the covariance of ${\bf H}$ itself. Also, the eigenvalues are limited to $n_Rn_S$ by physical principles. Hence,~\eqref{correlation_matrix} provides a significant computational saving compared to its direct computation, which would require knowledge of $(N_RN_S)^2$ real-valued parameters.

As shown in~\cite[Eq. (6)]{Weichselberger2006}, the correlation matrix ${\bf R}$ in~\eqref{correlation_matrix} has as a Kronecker structure on eigenmode level. If the separability condition is imposed, by inspection of \eqref{channel_matrix_kronecker} in Corollary~\ref{th:separable_model} and the statistical equivalence between $e^{\imagunit\boldsymbol{\Gamma}_R}{\vect{W}} e^{-\imagunit\boldsymbol{\Gamma}_S}$ and ${\vect{W}}$ it follows that
\begin{equation}
\vect{R} = \vect{R}_S \kron \vect{R}_R
\end{equation}
with correlation matrices $\vect{R}_R  = \boldsymbol{\Phi}_R \diag(\boldsymbol{\sigma}_R \odot \boldsymbol{\sigma}_R) \boldsymbol{\Phi}_R^{\Htran}  \in \Complex^{N_R \times N_R}$ and $\vect{R}_S  = \boldsymbol{\Phi}_S \diag(\boldsymbol{\sigma}_S \odot \boldsymbol{\sigma}_S) \boldsymbol{\Phi}_S^{\Htran}  \in \Complex^{N_S \times N_S}$ and eigenvalue matrix $\boldsymbol{\Lambda} = \diag(\boldsymbol{\sigma}_R \odot \boldsymbol{\sigma}_R) \otimes \diag(\boldsymbol{\sigma}_S \odot \boldsymbol{\sigma}_S)$.
Notice that, in the above Kronecker model, the correlation between two transmit (receive) antennas is the same irrespective of the receive (transmit) antenna where it is observed. This is not the case for the general coupled model in~\eqref{correlation_matrix}.

\begin{figure}
     \centering
     \begin{subfigure}[b]{\columnwidth}
         \centering
         \includegraphics[width=1\columnwidth]{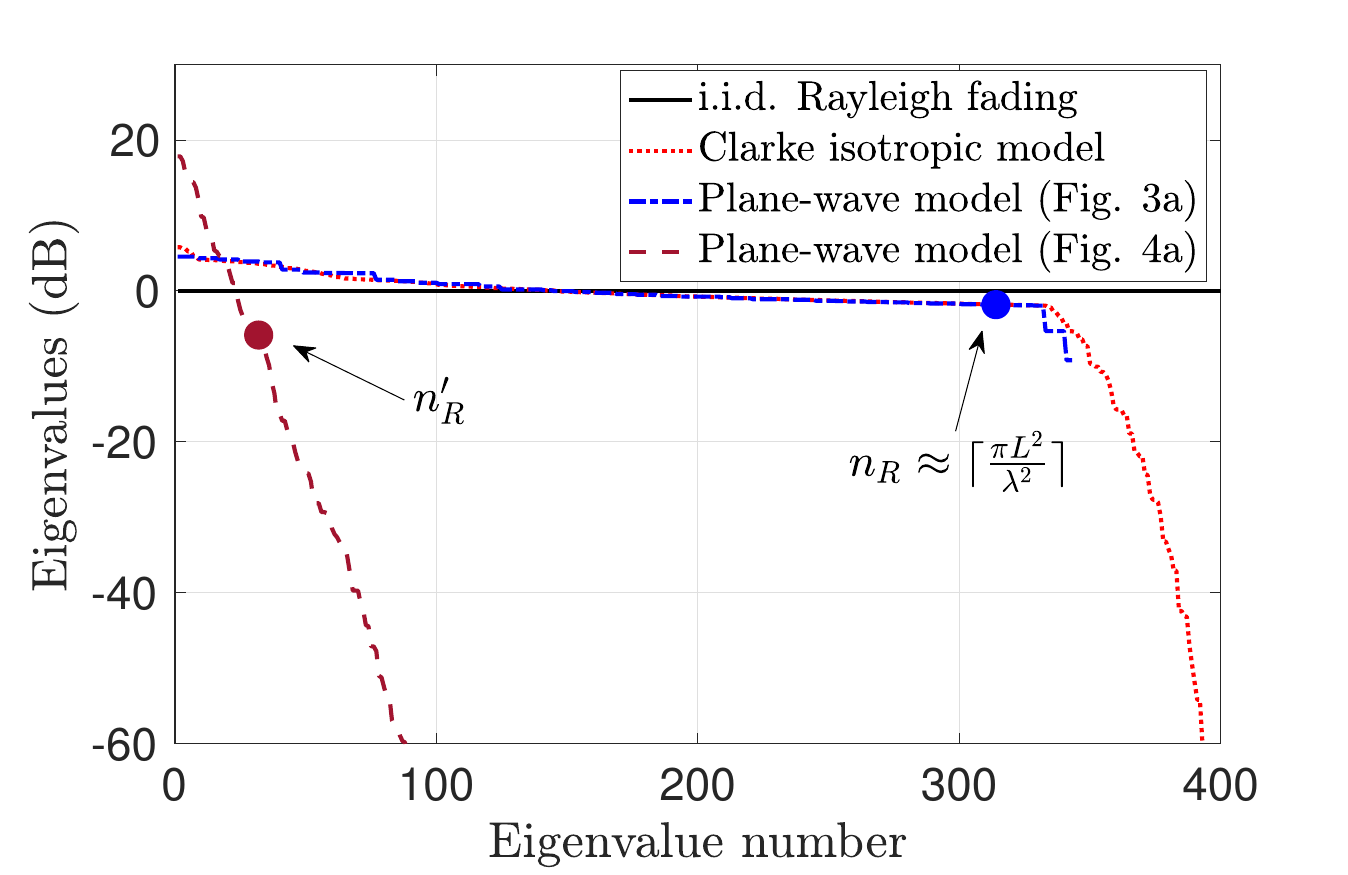}
         \caption{Array with $\lambda/2$ antenna spacing}
     \end{subfigure}
     \hfill
     \begin{subfigure}[b]{\columnwidth}
         \centering
         \includegraphics[width=1\columnwidth]{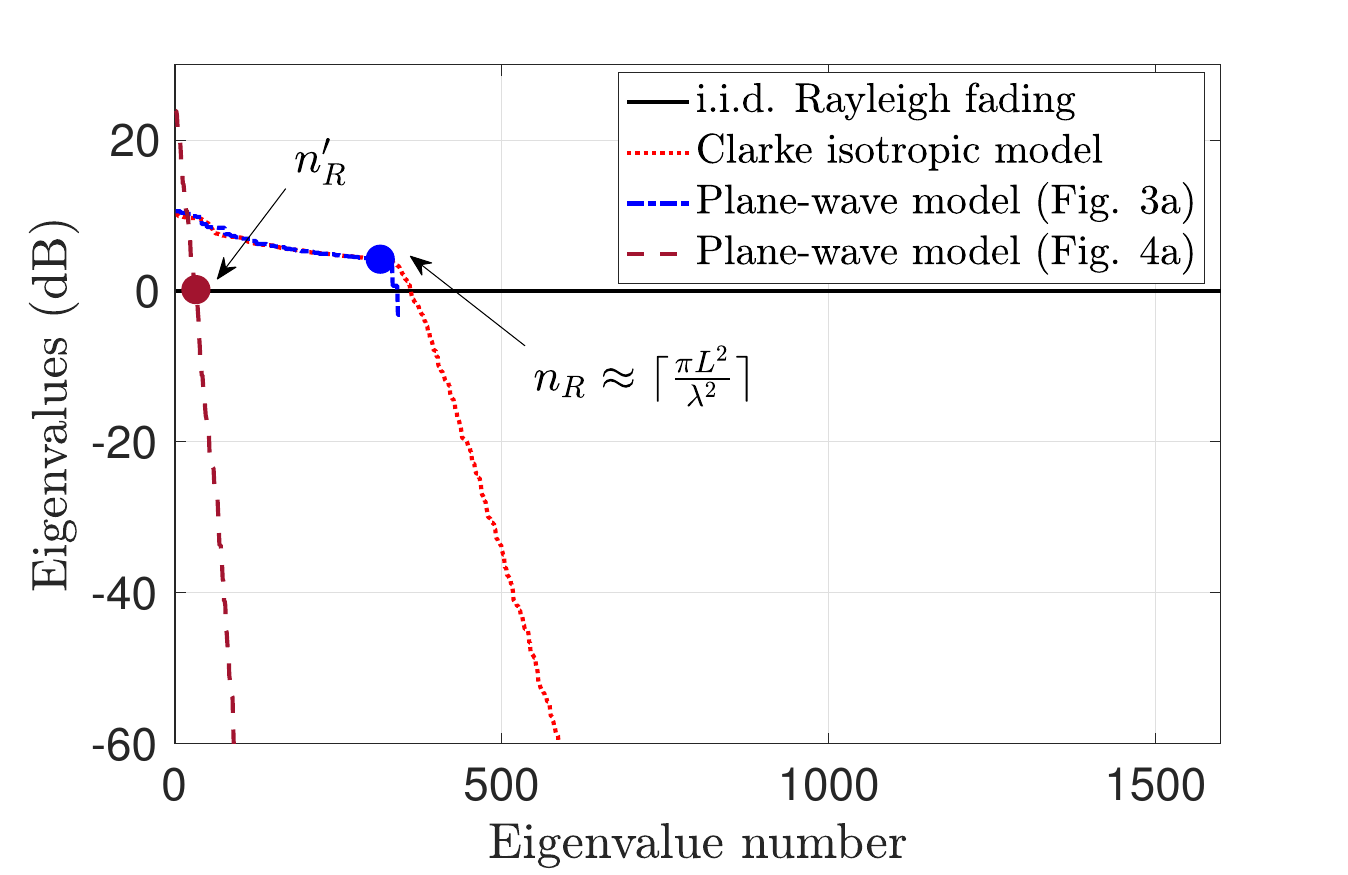}
         \caption{Array with $\lambda/4$ antenna spacing}
     \end{subfigure}
        \caption{{\small Channel eigenvalues (in dB) of $\vect{R}_R$ reported in a descending order for a squared array of size $L/\lambda=10$ in a setup with $\lambda/2$- and $\lambda/4$-spaced antenna elements (i.e., $N_R=400$ and $N_R=1600$). The Fourier plane-wave model for the scattering conditions in Fig.~\ref{fig:varIso}(a) and Fig.~\ref{fig:varNonIso}(a) is compared to the Clarke's isotropic model and i.i.d. Rayleigh fading model.}}
\label{fig:eigenvalues_lambda}
\end{figure}


The correlation properties of ${\bf H}$ depends on the eigenvalues of $\bf {R}$. For the Kronecker model in Corollary~\ref{th:separable_model} with symmetric scattering, we can concentrate only on the eigenvalues of $\vect{R}_R$. These are illustrated in Fig.~\ref{fig:eigenvalues_lambda} in dB sorted in a descending order for ${L/\lambda=10}$ and in a setup with $\lambda/2-$ and $\lambda/4-$spaced antenna elements (i.e., ${N_R=400}$ and ${N_R=1600}$), respectively.
Both isotropic and non-isotropic propagation conditions in Fig.~\ref{fig:varIso}(a) and Fig.~\ref{fig:varNonIso}(a) are considered. The number of significant coupling coefficients for the two cases is ${n_R^\prime=n_R\approx315}$ and ${n_R^\prime=36}$, as indicated by a circle on the corresponding curves. These determine the number of eigenvalues that carry the essential channel information.
Note that, for the i.i.d. Rayleigh fading model, we have $N_RN_S$ eigenvalues equal to $1$. Hence, the gap between the two is given by ${N_R}/{n_R} = {\lambda^2}/(\pi {\Delta_{R,x} \Delta_{R,y}})$ for a uniform antenna spacing of $\Delta_{R,x},\Delta_{R,y}$, which yields roughly a $1.2 \times$ and $5 \times$ overall increase in the number of eigenvalues. Remarkably, this error grows quadratically with the normalized antenna spacings.

\begin{figure*}[t!]
    \centering
     \includegraphics[width=1.5\columnwidth]{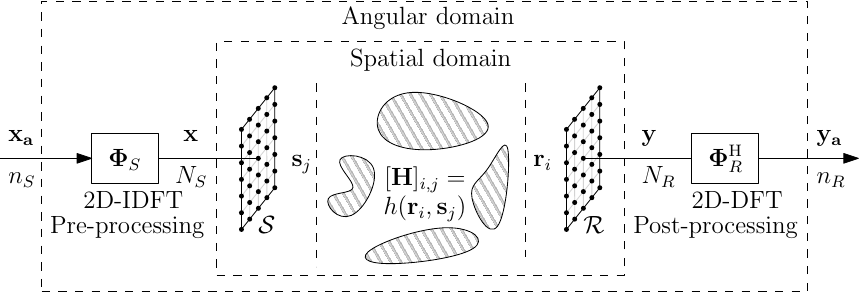} 
       \caption{{\small Transceiver architecture for communicating over the electromagnetic random MIMO channel.}}
   \label{fig:propagation_model}
\end{figure*}

As shown in Fig.~\ref{fig:eigenvalues_lambda}, the more uneven the coupling coefficients, the steeper the eigenvalues decay, which implies higher correlation. Ideally, if ${\boldsymbol{\Lambda} = N_SN_R \vect{I}_{n_S n_R}}$ in~\eqref{eigenvalues_channel} the channel samples would be mutually independent, thus leading to the i.i.d. Rayleigh fading model.
However, this is never the case.
In fact, the strengths of the coupling coefficients are not all equal even in the presence of isotropic propagation~\cite{PizzoJSAC20}, as it follows from Fig.~\ref{fig:varIso}. 
This proves that an electromagnetic random MIMO channel necessarily exhibits spatial correlation~\cite{PizzoJSAC20} and implies that the i.i.d. Rayleigh fading model shall never be used to model ${\bf H}$~\cite{bjornson2020rayleigh}. 
The closest physically-tenable model to an i.i.d. Rayleigh fading is the Clarke's isotropic model~\cite{PizzoJSAC20,MarzettaISIT}. This is generated from the Clarke's spatial correlation matrix whose $(i,j)-$th entry is $\sinc(2 d_{ij}/\lambda)$ where $d_{ij}$ is the distance between the $i-$th and $j-$th receive antennas.
As seen in Fig.~\ref{fig:eigenvalues_lambda}, the Fourier plane-wave model with isotropic propagation provides us with an $n_R-$order low-rank approximation of Clarke's model. The error due to ``discarding'' ${N_R-n_R}$ eigenvalues is approximately $4.6\%$ of the total channel power at $L/\lambda=10$. This reduces to $2.3\%$ when $L/\lambda=30$ and approaches zero asymptotically. In terms of capacity, for $L/\lambda=10$ we already obtain a high accuracy, as shown in \cite[Fig.~2]{PizzoASILOMAR20}.

\subsection{Channel generation and migration filters}
The generation of the electromagnetic MIMO channel in~\eqref{channel_matrix} requires only knowledge of the strength of the coupling coefficients~\eqref{variances_energy}, which are collected in the matrix $\boldsymbol{\Sigma}$. If this knowledge is available, the channel matrix $\vect{H}$ can be generated as follows: \emph{i}) Generate $\vect{W} \in \Complex^{n_R \times n_S}$ with independent entries $\CN(0,1)$; \emph{ii}) Compute the coupling matrix ${\vect{H}}_a$ in~\eqref{angular_coefficients}; \emph{iii}) Compute $\widetilde{\vect{H}}$ for any $s_z$ and $r_z > s_z$ as specified in Lemma~\ref{th:channel_matrix_lemma}; \emph{iv}) Obtain ${\vect{H}}$ in~\eqref{channel_matrix} as $\boldsymbol{\Phi}_R \big( \boldsymbol{\Phi}_S\widetilde{\vect{H}}^{\Htran}\big)^{\Htran}$. With uniform sampling, the matrices $\boldsymbol{\Phi}_S$ and $\boldsymbol{\Phi}_R$ reduce to two 2D IDFT transforms. Hence, the last step has a relatively low complexity due to the use of the FFT (Fast Fourier Transform). Notice that the computation of ${\vect{H}}$ for a different pair $(r_z,s_z)$ requires only to perform the third and fourth steps. The third step requires left- and right-multiplication by the propagators (migrators) $e^{\boldsymbol{\Gamma}_R}$ and $e^{-\boldsymbol{\Gamma}_S}$ that are known a priori as they are fully determined by the array geometry and wavelength. 
 
 Alternatively, the MIMO channel can be generated by using the eigendecomposition of the spatial correlation matrix in~\eqref{correlation_matrix} as ${\rm vec}({\bf H}) = {\bf U} \boldsymbol{\Lambda}^{1/2} {\rm vec}(\vect{W})$. This way to generate the MIMO channel should be used whenever one is interested in metrics (such as the capacity), where only the statistical equivalence between $\widetilde{\vect{H}}$ and ${\vect{H}}_a$ matters. In this case, the application of propagator matrices becomes irrelevant.

{The Fourier plane-wave series expansion of the channel in Theorem~\ref{th:series_expansion} generates a periodic stationary random field that repeats exactly after a 2D period of dimensions $(L_{S,x},L_{S,y})$ and $(L_{R,x},L_{R,y})$ at source and receiver, respectively. A similar periodic behavior is observed for the autocorrelation function in~\eqref{correlation_karhunen}, which must be continuous at the endpoints of each period. Hence, the array sizes must be large enough for the correlation properties of the aperiodic channel to be preserved,  in agreement with Assumption~\ref{Assumption1}.}

\subsection{Measurements of coupling coefficients} \label{sec:measurement_coupling_coeff}

The deterministic matrices $\boldsymbol{\Phi}_R$ and $\boldsymbol{\Phi}_S$ depend only on the array geometries, and thus are known. This is a useful property to estimate the coupling coefficients of ${\bf H}$. Indeed, pilot signals can be transmitted along their vectors. Hence, it is sufficient to transmit approximately $n_Rn_S$ pilot signals. Notice that $n_R$ and $ n_S$ depend on the normalized array length with consequent increase of pilot resources needed for channel estimation (see Fig.~\ref{fig:varIso} and Fig.~\ref{fig:varNonIso}).

The strength of the coupling coefficients depends exclusively on the scattering mechanisms, which evolve slowly in time compared to the fast variations of $\vect{H}$. This implies that they can be estimated with high accuracy on the basis of $\vect{H}$. From~\eqref{Fourier_series_discrete_KL}, they are given by
\begin{equation}\label{variances_energy}
\!\!\!\!\sigma^2(\ell_x,\ell_y,m_x,m_y) =  \mathbb{E}\left\{\left|\boldsymbol{\phi}_R^{\Htran}(\ell_x,\ell_y)\vect{H}\boldsymbol{\phi}_S(m_x,m_y)\right|^2\right\}
\end{equation}
which provides a possible way to estimate without the need of measuring the channel matrix ${\bf H}$; that is, transmit the signal $\boldsymbol{\phi}_S(m_x,m_y)$, project the received signal onto $\boldsymbol{\phi}_R(\ell_x,\ell_y)$, compute the square of the absolute value of measurement and take the average over the different measurements. Recall that only the knowledge of the $n_Rn_S$ real-valued coefficients $\{\sigma^2(\ell_x,\ell_y,m_x,m_y)\}$ is needed to fully specify the correlation matrix $\bf {R}$ in~\eqref{correlation_matrix}. To validate the accuracy of the developed model, the correlation matrix $\bf {R}$ should be compared to the sample correlation matrix obtained from real-world measurements of $\vect{H}$.

\section{Capacity Evaluation}  \label{sec:capacity}

We now use the developed channel model to numerically evaluate the capacity of the MIMO communication system in~\eqref{eq:MIMO_channel}. From the semi-unitary equivalence between the spatial-domain and angular-domain in Lemma~\ref{th:channel_matrix_lemma}, we have that~\eqref{eq:MIMO_channel} is equivalent to
\begin{equation} \label{conv_model_noise_angular}
\vect{y}_a  =  \sqrt{{\rm snr}} \vect{H}_a \vect{x}_a + \vect{n}_a 
\end{equation}
where $\vect{y}_a = {\boldsymbol \Phi}_R^{\Htran} \vect{y} \in \Complex^{n_R}$ and $\vect{x}_a = {\boldsymbol \Phi}_S^{\Htran} \vect{x} \in \Complex^{n_S}$ denote the received and transmitted signal vectors in the angular domain, respectively. Here, ${\rm snr}$ is the signal-to-noise ratio (SNR) at the receiver that is comprehensive of the large-scale fading coefficient.
Also, $\vect{n}_a  = {\boldsymbol \Phi}_R^{\Htran} \vect{n} \in \Complex^{n_R}$ is the angular noise vector distributed as $\vect{n}_a \sim\CN({\bf 0}, {\bf I}_{n_R})$. Unlike~\eqref{eq:MIMO_channel}, the entry $[{\bf H}_a]_{ij}$ represents the coupling coefficient between the $j$th angular set $\mathcal{W}_S(m_x,m_y)$ in~\eqref{S_s} and the $i$th angular set $\mathcal{W}_R(\ell_x,\ell_y)$ in~\eqref{S_r}. Under the assumption $\vect{Q}_a = \Ex\{\vect{x}_a^{\Htran} \vect{x}_a\} \le 1$, the ergodic capacity of~\eqref{conv_model_noise_angular} in bit/s/Hz is
\begin{align} \label{ergodic_capacity}
C  =  \max_{\vect{Q}_a: \tr(\vect{Q}_a) \le 1}  \Ex\{ \log_2 \det\left( \vect{I}_{n_R} + {\rm snr} \vect{H}_a \vect{Q}_a \vect{H}_a^{\Htran}\right)
\}.
\end{align}
 This is computed and quantified next under different degrees of channel state information.

\begin{remark}The transceiver architecture for communicating over the electromagnetic MIMO channel is illustrated Fig.~\ref{fig:propagation_model}. As seen, the transformation from the angular-domain to the spatial-domain (and viceversa) is fully determined by the matrices $\boldsymbol{\Phi}_S$ and $\boldsymbol{\Phi}_R^{\Htran}$, which depend only on the array geometries. With uniform sampling, $\boldsymbol{\Phi}_S$ and $\boldsymbol{\Phi}_R^{\Htran}$ become two-dimensional DFT and IDFT matrices, which can be efficiently implemented in the analog domain by using a Butler matrix \cite{Molisch2004} or by means of lens antenna arrays \cite{ZhangLens}. Signal processing algorithms operate in the angular domain and thus their complexity depend on ${\bf H}_a$, i.e., $n_R$ and $n_S$ at maximum. Hence, we can operate in a regime where $N_S \gg n_S$ and $N_R \gg n_R$ without any impact on the signal processing (e.g., channel estimation, optimal signaling, coding).
\end{remark}

\begin{figure}[t!]
    \centering
     \includegraphics[width=1.1\columnwidth]{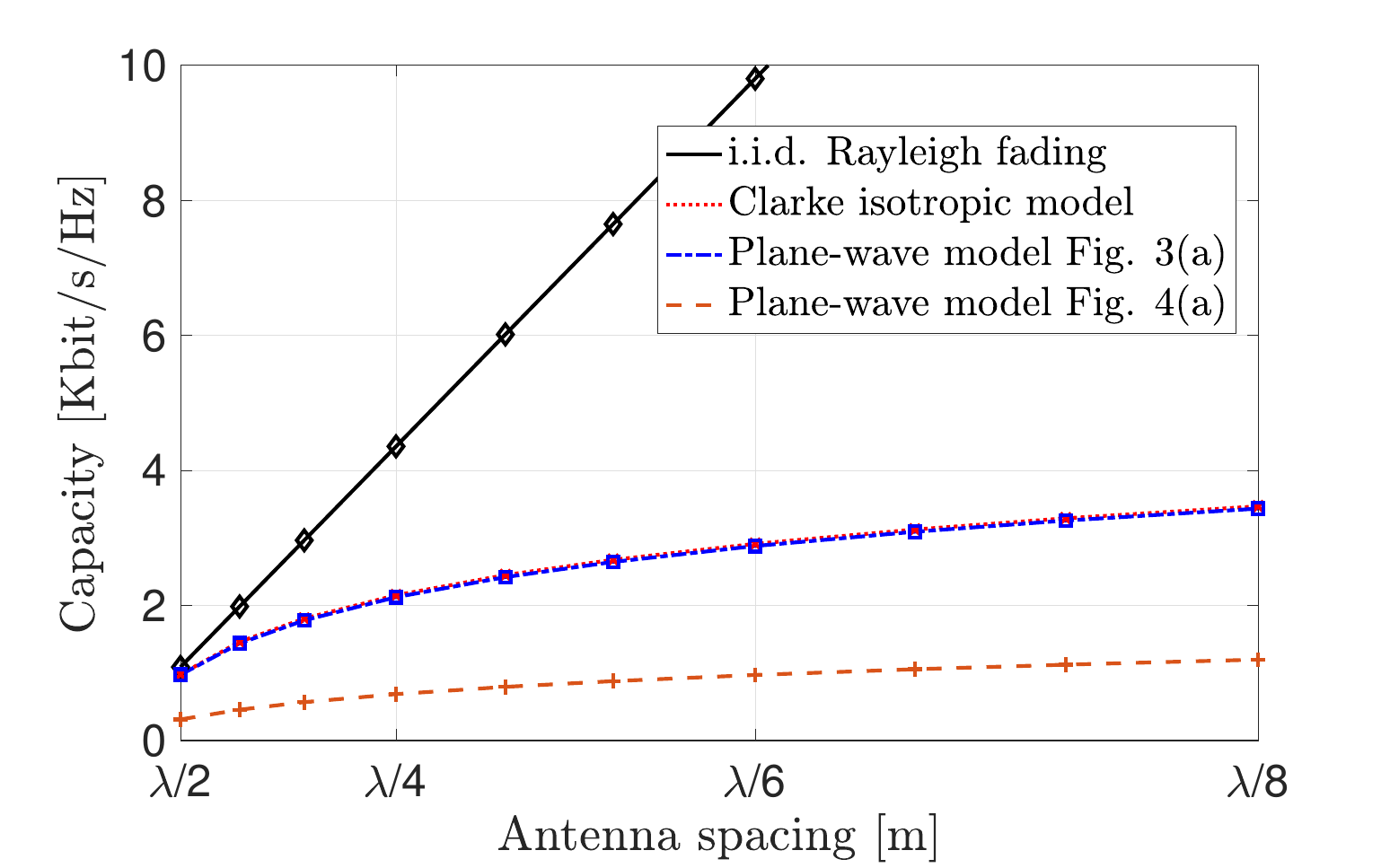} 
  \caption{{\small Ergodic capacity in Kbit/s/Hz as a function of antenna spacing $\Delta\in[\lambda/2,\lambda/8]$ for a squared array of size $L/\lambda=10$ and ${\rm snr} = 10$~dB. The Fourier plane-wave model is generated by using the setting of Fig.~\ref{fig:varIso}(a) and Fig.~\ref{fig:varNonIso}(a).}}
   \label{fig:capacity_NoCSI}
\end{figure}
\subsection{Perfect channel knowledge at receiver}

With instantaneous channel state information available at the receiver, the ergodic capacity in~\eqref{ergodic_capacity} is achieved by an i.i.d. input vector $\vect{x}_a$ with $\vect{Q}_a = \frac{1}{n_S}{\bf I}_{n_S}$ \cite{heath_lozano_2018} and is given by
\begin{align}\label{capacity_3}
C =\sum_{i=1}^{{\rm rank}(\vect{H}_a)}\mathbb{E}\left\{\log_2 \left(1 + \frac{{\rm snr} }{n_S} \, \lambda_i\left(\vect{H}_a\vect{H}_a^{\Htran}\right)\right) \right\} 
\end{align}
where $\{\lambda_i\left(\vect{A}\right)\}$ are the eigenvalues of an arbitrary $\vect{A}$. Under separability, \eqref{capacity_3} reduces to
\begin{align} \notag
C = & \sum_{i=1}^{{\rm rank}(\vect{H}_a)}\mathbb{E}\bigg\{\log_2 \Big(1 + \frac{{\rm snr}}{n_S} \times\\&\, \times\lambda_i\big(\diag(\boldsymbol{\sigma}_R \odot \boldsymbol{\sigma}_R)  \vect{W} \vect{W}^{\Htran}
\diag(\boldsymbol{\sigma}_S \odot \boldsymbol{\sigma}_S)\big)\Big) \bigg\}\label{capacity_separable}
\end{align}
as obtained plugging \eqref{matrix_angular_coefficients_kronecker} into \eqref{capacity_3}.
Under Assumption~\ref{Assumption1}, $n_S,n_R \gg 1$ and tools from random matrix theory can be used to asymptotically approximate~\eqref{capacity_separable} as~\cite[Eq.~(102)]{Tulino2005}
\begin{align} \notag
C \approx \sum_{j=1}^{n_S}\log_2&\left(\frac{1+{\rm snr} [\boldsymbol{\sigma}_S \odot \boldsymbol{\sigma}_S]_j \Gamma_R}{e^{{\rm snr}\Gamma_S\Gamma_R}}\right) \\&+ \sum_{i=1}^{n_R}\log_2\left(1+ {\rm snr} [\boldsymbol{\sigma}_R \odot \boldsymbol{\sigma}_R]_i \Gamma_S\right)\label{capacity_asym}
\end{align} 
where the coefficients $\Gamma_R,\Gamma_S$ are obtained by solving the fixed-point equations \cite[Eqs.~(103)--(104)]{Tulino2005}:
\begin{align}
\Gamma_R &= \frac{1}{n_S}\sum_{i=1}^{n_R} \frac{[\boldsymbol{\sigma}_R \odot \boldsymbol{\sigma}_R]_i}{1 + {\rho} [\boldsymbol{\sigma}_R \odot \boldsymbol{\sigma}_R]_i\Gamma_S} \\
\Gamma_S &= \frac{1}{n_S}\sum_{j=1}^{n_S} \frac{[\boldsymbol{\sigma}_S \odot \boldsymbol{\sigma}_S]_j}{1 + {\rho} [\boldsymbol{\sigma}_S \odot \boldsymbol{\sigma}_S]_j \Gamma_R}.
\end{align}
In Fig.~\ref{fig:capacity_NoCSI}, we plot the ergodic capacity in Kbit/s/Hz as a function of antenna spacing $\Delta\in[\lambda/2,\lambda/8]$ in the same setup of Fig.~\ref{fig:eigenvalues_lambda}(a) with $L/\lambda=10$.
The continuous lines are generated from~\eqref{capacity_separable} by using Monte Carlo simulations, whereas markers are generated according to the large-dimensional approximation in~\eqref{capacity_asym}. A perfect match is observed although $n_S,n_R$ are finite. Comparisons are made with the Clarke's model in which power is allocated only onto the most significant $n_S$ eigenmodes. The perfect match with the Fourier plane-wave model validates our physical low-rank approximation under isotropic propagation conditions. The capacity with i.i.d. Rayleigh fading is also reported as reference. Compared to this model, a large gap is observed for $\Delta < \lambda/2$ due to the correlation that naturally arises among antennas when $\Delta$ decreases. This confirms that i.i.d. Rayleigh fading is highly inadequate to model the channel with planar arrays of sub-wavelength spacing. In fact, it cannot be derived from physic principles when planar arrays are considered~\cite{PizzoJSAC20}. The error in terms of DoF is inversely proportional to the square of normalized antenna spacing.
Note that the capacity per transmitted stream of information is given by $C/\min(n_R,n_S)$. In a setup with ${\Delta = \lambda/2}$, it is approximately equal to {$3.4$ and $2.8$}~bit/s/Hz under the isotropic and non-isotropic scenarios described in Fig.~\ref{fig:varIso}(a) and Fig.~\ref{fig:varNonIso}(a), respectively.

\subsection{Perfect channel knowledge at source and receiver}

\begin{figure}[t!]
    \centering
     \includegraphics[width=1.1\columnwidth]{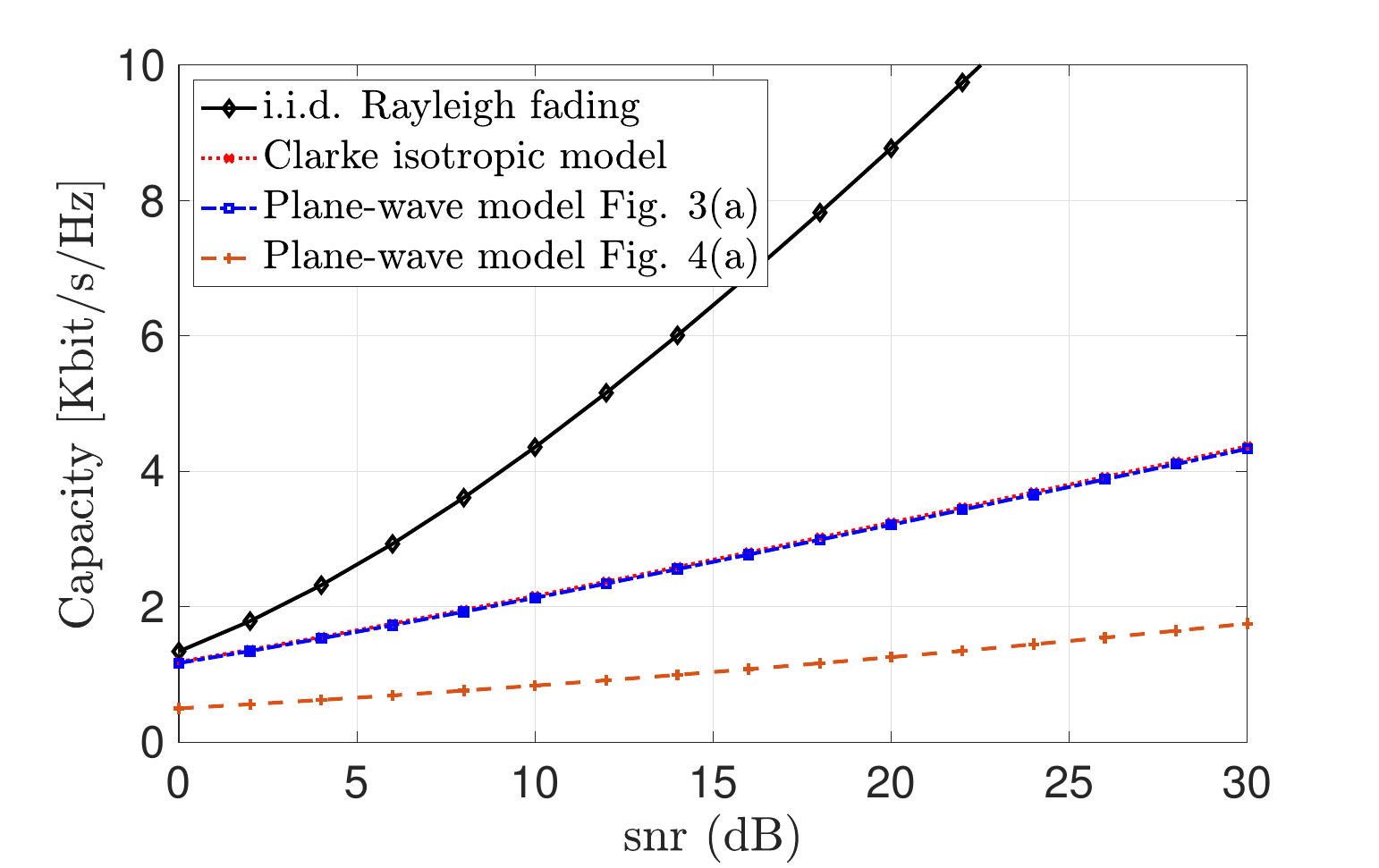}
  \caption{{\small Ergodic capacity $C$ in Kbit/s/Hz as a function of ${\rm snr}$ in~dB for $L/\lambda=10$ with $\lambda/4$-spaced antenna elements (i.e., $N_R=1600$). The Fourier plane-wave model is generated in the setups of Fig.~\ref{fig:varIso}(a) and Fig.~\ref{fig:varNonIso}(a).}}
   \label{fig:capacity_CSI_lambda4}
\end{figure}

Assume $\vect{H}_a$ is perfectly known at both sides and let ${\vect{H}_a = \vect{U}_a \boldsymbol{\Lambda}_a \vect{V}_a^{\Htran}}$ be its singular value decomposition. The capacity in~\eqref{ergodic_capacity} is achieved by a circularly-symmetric, complex-Gaussian angular input ${\vect{x}_a = \vect{V}_a \vect{P}_a^{1/2} \vect{s}_a}$ where $\vect{s}_a\in\Complex^{n_S}$ is an i.i.d. circularly-symmetric complex-Gaussian vector with unit variance and $\vect{P}_a \in \Complex^{n_S \times n_S}$ is diagonal with entries given by the optimal powers, computed via the waterfilling algorithm (e.g.,~\cite{heath_lozano_2018}). Hence, $\vect{x}_a$ has covariance matrix $\vect{Q}_a = \vect{V}_a \vect{P}_a \vect{V}_a^{\Htran}$ and the ergodic capacity is
\begin{align} \label{capacity_SVD}
C =\sum_{i=1}^{{\rm rank}(\vect{H}_a)} \log_2 \left( \mu \lambda_i(\vect{H}_a \vect{H}_a^{\Htran}) \right)^+
\end{align}
where $\mu$ is such that ${\rm snr} = \sum_{i=1}^{{\rm rank}(\vect{H}_a)} \left(\mu - {1}/{\lambda_i(\vect{H}_a \vect{H}_a^{\Htran})} \right)^+$. Fig.~\ref{fig:capacity_CSI_lambda4} plots~\eqref{capacity_SVD} as a function of ${\rm snr}$ in dB. As ${\rm snr}$ grows large, an increasing number of communications modes is activated. In the large SNR regime, the capacity scales linearly in $\log_2({\rm snr})$ with slope given by the number of DoF in~\eqref{DoF}.
Compared to i.i.d. Rayleigh fading and Clarke's model, the spatial correlation reduces the slope and introduces a negative shift in the capacity~\cite{heath_lozano_2018}. 

\subsection{Statical knowledge of the channel at the source}
The key message from the analysis in Section IV is that knowledge of the strength of coupling coefficients is needed to have full statistical knowledge of $\vect{H}_a$ in~\eqref{angular_coefficients}. A possible way to obtain this information is sketched in Section~\ref{sec:measurement_coupling_coeff}. Suppose now that this knowledge is perfectly available at the transmitter. Since $\vect{H}_a$ has independent entries whose marginal distributions are symmetric with respect to zero, the optimal angular covariance matrix is diagonal, i.e., ${\vect{Q}_a = \vect{P}_a}$ \cite[Th.~1]{Veeravalli}. 
The optimal $\vect{x}_a$ is ${\vect{x}_a = \vect{P}_a^{1/2} \vect{s}_a}$ where the information-bearing vector ${\vect{s}_a \in \Complex^{n_S}}$ is an i.i.d. circularly-symmetric complex-Gaussian vector with unit variance.
Hence, the capacity-achieving transmission strategy is to send statistically-independent streams of information angularly. 
Unlike the angular domain, in the spatial domain we have statistically correlated input symbols specified by the correlation matrix ${\vect{Q} = {\boldsymbol \Phi}_S \vect{Q}_a {\boldsymbol \Phi}_S^{\Htran} \in \Complex^{N_S\times N_S}}$. 

\section{Conclusions} \label{sec:conclusions}

We introduced a novel Fourier plane-wave stochastic channel model that is mathematically tractable and consistent with the physics of wave propagation. 
The developed model is even valid in the near-field and fully captures the essence of electromagnetic propagation under arbitrary scattering conditions. It is especially for, but not limited to, conducting research on future wireless systems with {electromagnetically large} and dense antenna arrays.  
Our hope is to excite the interest of the wireless research community toward the {development of physics-inspired models} that may push further the limits of MIMO communications \cite{Migliore_EM}.
An important extension of the proposed Fourier plane-wave model is the incorporation of polarized antenna arrays~\cite{MarzettaVector,MarzettaNOKIA} and also of the mutual coupling between antenna elements, which may critically affect the performance of dense arrays (e.g., \cite{Laas2020}).
{Real-world measurements are needed to support the developed theory by correctly extracting model parameters for a prescribed environmental class.}

\appendices

\section*{Appendix} \label{app:fourier_series}
We aim to provide a discrete approximation of the Fourier plane-wave representation in~\eqref{Fourier_planewave} for a channel observed over a large spatial region of finite extent, as specified in Assumption~1. To this end, we follow the same approach exemplified in \cite[App.~IV.A]{PizzoJSAC20} for a 1D time-domain random process and partition the integration region $\mathcal{D}(\kappa)\times \mathcal{D}(\kappa)$ of $h(\vect{r},\vect{s})$ uniformly with angular sets $\mathcal{W}_S(m_x,m_y)$ and $\mathcal{W}_R(\ell_x,\ell_y)$ (see \cite[Eq.~(61)]{PizzoJSAC20}):
\begin{align} \notag
&\frac{1}{(2\pi)^2} \mathop{\sum}_{\ell_x,\ell_y}  \mathop{ \sum}_{m_x,m_y}  \!\!\!
 \iiiint_{\mathcal{W}_S(m_x,m_y) \mathcal{W}_{R}(\ell_x,\ell_y)} \hspace{-2.5cm}
S^{1/2}(k_x,k_y,\kappa_x,\kappa_y)  \\ & \hspace{.3cm} \times W(k_x,k_y,\kappa_x,\kappa_y) a_R(\krx,\vect{r})
  a_S(\ktx,\vect{s}) \, dk_xdk_y d\kappa_xd\kappa_y\label{Fourier_planewave_partition}
\end{align}
where we used \eqref{scattering_response_nlos} with $\mathcal{W}_S(m_x,m_y)$ and $\mathcal{W}_R(\ell_x,\ell_y)$ being given by~\eqref{S_s} and \eqref{S_r}, respectively.
The application of the first mean-value theorem over each partition yields the approximated Fourier series expansion in~\eqref{Fourier_series} where each random coefficient $H_a(\ell_x,\ell_y,m_x,m_y)$ is given by
\begin{align}  \notag
\frac{1}{(2\pi)^2} &\iiiint_{\mathcal{W}_S(m_x,m_y) \times\mathcal{W}_R(\ell_x,\ell_y)}\hspace{-2.5cm}
S^{1/2}(k_x,k_y,\kappa_x,\kappa_y)  \\ & \hspace{.8cm} \times W(k_x,k_y,\kappa_x,\kappa_y) 
    dk_xdk_y d\kappa_xd\kappa_y\label{H_a_coeff}
\end{align}
for ${(m_x,m_y)\in\mathcal{E}_S}$ and ${(\ell_x,\ell_y)\in\mathcal{E}_R}$. 
Since these coefficients are obtained by projecting a 4D white-noise complex-Gaussian field $W(k_x,k_y,\kappa_x,\kappa_y)$ over a set of orthonormal functions, they are mutually-independent and circularly-symmetric, complex-Gaussian random variables \cite{PizzoJSAC20}. Their variances $\sigma^2(\ell_x,\ell_y,m_x,m_y)$ are obtained by computing the average power in \eqref{variances}, as shown next. 

\begin{figure}[t!]
    \centering
     \includegraphics[width=1\columnwidth]{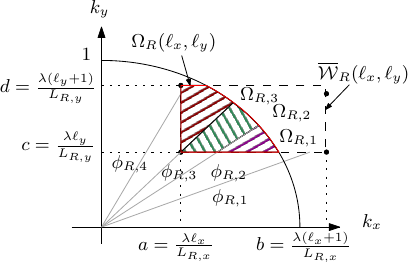}
  \caption{{\small Integration region $\Omega_R(\ell_x,\ell_y)$ in~\eqref{variances_5} for $\ell_x,\ell_y>0$.}}
   \label{fig:integration_subregion}
\end{figure}
Unlike \cite[App.~IV.C]{PizzoJSAC20}, where we were able to compute a closed-form expression for the variances under isotropic scattering, here we resort to a general numerical procedure as solving \eqref{variances} for every possible non-isotropic propagation conditions is pointless. Plugging \eqref{psd_4d} into \eqref{variances},
\begin{align}  \notag
 \iiiint_{\mathcal{W}_S(m_x,m_y) \times \mathcal{W}_{R}(\ell_x,\ell_y)} 
\hspace{-2.5cm}  & \mathbbm{1}_{\mathcal{D}(\kappa)}(k_x,k_y) \mathbbm{1}_{\mathcal{D}(\kappa)}(\kappa_x,\kappa_y) 
  \\ & \hspace{.8cm} \times \frac{A^2(k_x,k_y,\kappa_x,\kappa_y)}{\gamma(k_x,k_y) \gamma(\kappa_x,\kappa_y)}  dk_xdk_y d\kappa_xd\kappa_y\label{variances_1} 
\end{align}
where $ \mathbbm{1}_{\mathcal{D}(\kappa)}(\cdot)$ accounts for the circularly-bandlimited support of the channel and the proportionality constants are embedded into the spectral factor $A^2(k_x,k_y,\kappa_x,\kappa_y)$ to meet the unit average power constraint. After rescaling the integration domain in \eqref{variances_1} by $\kappa=2\pi/\lambda$, we obtain
\begin{align}  \notag
 \iiiint_{\overline{\mathcal{W}}_s(m_x,m_y) \overline{\mathcal{W}}_R(\ell_x,\ell_y)} 
\hspace{-2.5cm}  &\mathbbm{1}_{\mathcal{D}(1)}({k}_x,{k}_y) \mathbbm{1}_{\mathcal{D}(1)}({\kappa}_x,{\kappa}_y)
 \times \\ & \hspace{.8cm} \times \frac{A^2({k}_x,{k}_y,{\kappa}_x,{\kappa}_y)}{{\gamma}({k}_x,{k}_y) {\gamma}({\kappa}_x,{\kappa}_y)}  d{k}_x d{k}_y d{\kappa}_x d{\kappa}_y\label{variances_2} 
\end{align}
where $\overline{\mathcal{W}}_S(m_x,m_y)$ and $\overline{\mathcal{W}}_R(\ell_x,\ell_y)$ are the normalized angular sets obtained from ${\mathcal{W}}_S(m_x,m_y)$ in \eqref{S_s} and ${\mathcal{W}}_R(\ell_x,\ell_y)$ in \eqref{S_r}.
The integration variables coincide to the cosine directions that specify every transmit $\hat\ktx$ and receive $\hat\krx$ propagation directions.
The receive integration region $\bar{\mathcal{W}}_R(\ell_x,\ell_y)$ in \eqref{variances_2} is illustrated in Fig.~\ref{fig:integration_subregion} for the first wavenumber quadrant only, that is, $(\ell_x,\ell_y)\in\mathcal{E}_R$ such that $\ell_x,\ell_y>0$.
Due to the rotational symmetry of $\gamma(\cdot,\cdot)$ in \eqref{variances_2} we change integration variables to polar wavenumber coordinates $({k}_x,{k}_y) = ({k}_r \cos{k}_\phi, {k}_r \sin{k}_\phi)$ with ${k}_r \in [0,1]$ and ${k}_\phi\in[0,2\pi)$:
\begin{align}  \notag
 \iiiint_{\overline{\mathcal{W}}_S(m_x,m_y) \times \overline{\mathcal{W}}_{R}(\ell_x,\ell_y)} 
\hspace{-2.5cm}  &\mathbbm{1}_{[0,1]}({k}_r) \mathbbm{1}_{[0,1]}({\kappa}_r)  \\ & \times
 \frac{A^2({k}_r,{k}_\phi,{\kappa}_r,{\kappa}_\phi) {k}_r {\kappa}_r}{\sqrt{1-{k}_r^2} \sqrt{1-{\kappa}_r^2}}  d{k}_r d{k}_\phi d{\kappa}_r d{\kappa}_\phi.\label{variances_3} 
\end{align}
Typically, the field's directionality is expressed in spherical coordinates, i.e., elevation $(\theta_S,\theta_R) \in [0,\pi]$ and azimuth $(\phi_S,\phi_R)\in[0,2\pi)$ angles through the spectral factor $A^2(\theta_R,\phi_R,\theta_S,\phi_S)$.
The map between wavenumber coordinates and spherical coordinates is ${k_x = \sin\theta_R \cos \phi_R}$ and ${k_y = \sin\theta_R \sin\phi_R}$, which substituted into \eqref{gamma} yields ${\gamma(k_x, k_y) = \cos \theta_R}$, e.g., at receiver. 
The polar wavenumber coordinates follow directly as $k_r = \sin\theta_R$ and $k_\phi = \phi_R$
with Jacobian given by $|{\partial( k_r, k_\phi)}/{\partial(\theta_R,\phi_R)}| = \cos\theta_R$. 
In doing so, the terms at the denominator of \eqref{variances_3} disappear as they are embedded into the Jacobian of this transformation.
Notice that the circularly-bandlimited constraint in~\eqref{variances_3} implies that $\theta_R \in[0,\pi/2]$.
In other words, we consider propagation directions $\krx$ defined over the unit upper hemisphere,\footnote{For the unit lower hemisphere, replace $-\gamma( k_x,  k_y)$ with $\gamma( k_x,  k_y)$, which leads to $\theta_R \in(\pi/2,\pi]$. Physically, this corresponds to a propagation scenario with scatterers located behind the receiver \cite{PizzoIT21}.} which leads to
\begin{align} \notag
\iiiint_{\overline{\mathcal{W}}_S(m_x,m_y) \times \overline{\mathcal{W}}_R(\ell_x,\ell_y)} 
\hspace{-2.5cm}  & \mathbbm{1}_{[0,\pi/2]}(\theta_R)  \mathbbm{1}_{[0,\pi/2]}(\theta_S)  \\ & \hspace{.8cm}  \times A^2(\theta_R,\phi_R,\theta_S,\phi_S)   \, d\Omega_S d\Omega_R\label{variances_4} 
\end{align}
where $d\Omega_R = \sin \theta_R d\theta_R d\phi_R$ and $d\Omega_S = \sin \theta_S d\theta_S d\phi_S$ are the differential element of solid angles pointed by $\hat \ktx$ and $\hat \krx$.
The above formula can be compactly rewritten as
\begin{align}\notag
\sigma^2&(\ell_x,\ell_y,m_x,m_y)  \\& = \iiiint_{\Omega_S(m_x,m_y) \times \Omega_R(\ell_x,\ell_y)} 
\hspace{-1.5cm} A^2(\theta_R,\phi_R,\theta_S,\phi_S)   \, d\Omega_S d\Omega_R \label{variances_5} 
\end{align}
where $\Omega_R(\ell_x,\ell_y)$ is the intersection of set $\overline{\mathcal{W}}_R(\ell_x,\ell_y)$ and the projected upper hemisphere, e.g., at the receiver (see Fig.~\ref{fig:integration_subregion}).
Next, we express the integration region $\Omega_R(\ell_x,\ell_y)$ as a function of the spherical angles $(\theta_R,\phi_R)$ for all $(\ell_x,\ell_y)\in \mathcal{E}_R$. A similar procedure should be considered for the source region.
Let $a = \lambda \ell_x/L_{R_x}$, $b = \lambda (\ell_x+1)/L_{R_x}$, $c = \lambda \ell_y/L_{R_y}$, and $d = \lambda (\ell_y+1)/L_{R_y}$ be the cosine directions evaluated in correspondence of the four vertices of $\overline{\mathcal{W}}_R(\ell_x,\ell_y)$ in Fig.~\ref{fig:integration_subregion}.
These divide the integration region $\Omega_R(\ell_x,\ell_y)$ into three subregions:  $\phi_R\in[\phi_{R,1},\phi_{R,2}]$, $\phi_R\in[\phi_{R,2},\phi_{R,3}]$, and $\phi_R\in[\phi_{R,3},\phi_{R,4}]$, which are limited by the azimuth angles $\phi_{R,1} < \phi_{R,2} < \phi_{R,3} < \phi_{R,4}$. 
Hence, \eqref{variances_5} can be rewritten as
\begin{align} \notag
\sigma^2&(\ell_x,\ell_y,m_x,m_y) \\& = \sum_{i=1}^3 \sum_{j=1}^3 \iiiint_{\Omega_{R,i}(\ell_x,\ell_y) \times \Omega_{S,j}(m_x,m_y)}
\hspace{-2.5cm}  A^2(\theta_R,\phi_R,\theta_S,\phi_S)   \, d\Omega_S d\Omega_R\label{variances_6}
\end{align}
where $\Omega_{R,i}(\ell_x,\ell_y)  = \{\theta_R\in [\theta_{R,{\rm min}}(\phi_R), \theta_{R,{\rm max}}(\phi_R)], \phi_R \in[\phi_{R,i},\phi_{R,i+1}]\}$. The integration regions are function of the fourth orthants and are not reported here due to space limitation.

\bibliographystyle{IEEEbib}
\bibliography{IEEEabrv,refs}

\end{document}